\documentclass[journal,onecolumn]{IEEEtran}
\newtheorem{theorem}{Theorem}
\newtheorem{lemma}{Lemma}
\newtheorem{corollary}{Corollary}

\ifCLASSINFOpdf
\else
\fi
\usepackage{amsmath}
\usepackage{amssymb}
\usepackage{bbm}
\begin{document}
%
\title{The Capacity of the Relay Channel}
%
%
%

\author{Jonathan Ponniah}
\maketitle

\begin{abstract}
The capacity of the relay channel is characterized by three schemes: decode-forward, coordinated compress-forward, and uncoordinated compress-forward.  The proof relies entirely on properties of typical sequences. 
\end{abstract}


%
\IEEEpeerreviewmaketitle

\section{Introduction}
\label{sec:introduction}
The capacity of the relay channel \cite{Van_Der_Meulen_1971} has long thought to be a problem that defies ordinary means of characterization.  We find the capacity of the relay channel using properties of typical sequences exclusively.  One property in particular, compactness, is especially useful.  The idea is that any arbitrary sequence of coding schemes lives in a compact space, so there will be a point in the space that the sequence of schemes returns to over and over again.  Even though the schemes in this sequence are totally arbitrary, we use this point to find regularity.  If the sequence of schemes is performing well in the limit, then it must be performing well at the point it keeps returning to.  By focusing on this point, we derive a ``compression'' cut-set bound on compress-forward schemes using a variation of the classical sphere-packing argument.  Our variation derives the codeword rate contributing significantly to the the mass of a compression.  We also find that the compress-forward space is is defined by three levers (or distributions): the reverse-compression channel, the input, and the compression prior.  These distributions are linked through the channel.  Out of this linkage emerges a zero-one circuit consistency law that determines whether or not a compress-forward scheme that purports to support coordination between the source and relay is actually feasible.

In Section \ref{sec:typicalsequences}, we present properties of typical sequences.  First, we look at sequences of conditionally-typical sets anchored to arbitrary sequences of codewords where the codewords are increasing in length.  The convergence properties we derive uniformly apply to all codeword sequences.  Second, we introduce the notion of sharp-typicality to get tight asymptotic bounds on the normalized exponent of the volumes of these sets.

In Section \ref{sec:pointtopointchannel}, we prove the converse to the capacity of the point-to-point channel.  Our approach is a sphere-packing argument.  We introduce the notion of significant types to justify assigning a single type to all codewords in a codebook.  Then we use the results in Section \ref{sec:typicalsequences} to find the minimum volume of a decoding set when the probability of decoding error is vanishing.  We get a bound on the codebook rate, but this bound is anchored to the sequence of codebook types.  We use the compactness of the simplex to find a subsequence of types that converge to a point on the simplex.

In Section \ref{sec:therelaychannel}, we characterize the capacity of the relay channel via three schemes: decode-forward, coordinated-compress-forward, and uncoordinated-compress-forward.  This requires some maneuvering around the ways the levers of the general compress-forward scheme interact with each other through the channel.  We identify a dichotomy in the compress-forward regime where only one degree-of-freedom exists between the input distribution and the reverse compression channel distribution.

In Section \ref{sec:proofoftheorem1} we give a proof of the capacity of the relay channel.  To start, we focus on codebooks with a block encoding/decoding structure.  We extend the notion of significant types to the joint source-relay codebook and take a subsequence where the joint type converges to a point on the simplex.  Then we further focus on good codebooks that fill up the space of relay observations.  The way we define good codebooks relies on the notion of sharp typicality in Section \ref{sec:typicalsequences}.  We use a genie argument \cite{Kramer2001} to isolate the source$\rightarrow$(relay, destination) channel from the relay$\rightarrow$destination channel.  This sets the stage for the main argument - deriving the compression cut-set-bound.  Our approach is to bound the rate the source codewords contributing to the mass of a compression.  A similar mass argument appears in \cite{Wu2009}.  We use the compactness of the interval defined by the cardinality of the relay symbol space (i.e., $[0,|{\cal X}_{r}|]$) to hone in on a well-behaved compression space.  We use a sphere-packing argument to get the compression cut-set bound.  

Finally, we tidy-up the assumptions. We justify our focus on good block-encoding schemes by using the compactness of the interval $[0,|{\cal X}_{r}|]$ to identify a block-coding regime and a symbol-by-symbol regime.  We also introduce additional constraints to swap out the genie.  We show the dichotomy in the compress-forward regime is due to an implicit zero-one circuit-consistency law that requires alignment between the forward channel distribution and the reverse-compression-channel distribution.

On the achievability side we introduce the coordinated compress-forward scheme.  The innovation here is partially encoding a message in the compression.  In \cite{Cover1979}, three schemes are proposed: decode-forward, compress-forward, and partial-decode-forward.  Partial-decode-forward does not achieve the compression cut-set bound.  Uncoordinated compress-forward is a simplified version of the compress-forward scheme in \cite{Cover1979} with no Wyner-Ziv coding.  Section \ref{sec:conclusion} concludes the paper.

\section{Typical Sequences}
\label{sec:typicalsequences}
We will establish some definitions/notation.  We denote a finite-length source sequence as an ordered set of symbols from the symbol set ${\cal X}_{s}$:

\begin{align}
{\bf x}^{(n)}_{s}:=\{(i,x^{(i)}_{s}): x^{(i)}_{s}\in{\cal X}_{s}, 1\leq i\leq n\}.
\end{align}
For a symbol $x^{(i)}_{s}\in{\cal X}_{s}$ define: 
\begin{align}
\mathbbm{1}_{\{x^{(i)}_{s}=a\}}:=\begin{cases}
1 & x^{(i)}_{s}=a,\\
0 & \text{otherwise}.
\end{cases}
\end{align}
In the notation above $x^{(i)}_{s}$ is the $i^{th}$ symbol in the finite-length sequence ${\bf x}^{(n)}_{s}$.  Next, we define the count function which returns the number of times a specific symbol appears in a sequence.

\begin{align}
N(a|{\bf x}^{(n)}_{s})=\sum^{n}_{i=1}\mathbbm{1}_{\{x^{(i)}_{s}=a\}}.
\end{align}

We will also need to define the count function for a symbol pair $(a,b)\in{\cal X}_{s}\times{\cal Y}_{d}$ over the sequence pair $({\bf x}^{(n)}_{s},{\bf y}^{(n)}_{d})$.  We define the indicator function:

\begin{align}
\mathbbm{1}_{\{x^{(i)}_{s}=a,\hspace{0.5mm}y^{(i)}_{r}=b\}}:=\begin{cases}
1 & x^{(i)}_{s}=a\text{ and }y^{(i)}_{d}=b,\\
0 & \text{otherwise}.
\end{cases}
\end{align}
Now we define the associated count function:

\begin{align}
N((a,b)|({\bf x}^{(n)}_{s}, {\bf y}^{(n)}_{d}))=\sum^{n}_{i=1}\mathbbm{1}_{\{x^{(i)}_{s}=a,\hspace{0.5mm}y^{(i)}_{d}=b\}}.
\end{align}
Next, for a fixed $\epsilon_{X}>0$, we define strong conditionally typical sequences:

\begin{align}
    \label{def:strongtypicalityv1}
    T^{(n)}_{\epsilon_{X}}(Y_{d}|{\bf x}^{(n)}_{s}):=\left\{{\bf y}^{(n)}_{d}:\frac{1}{n}|N((a,b)|({\bf x}^{(n)}_{s},{\bf y}^{(n)}_{d}))-p(b|a)N(a|{\bf x}^{(n)}_{s})|<\epsilon_{X}\right\},\hspace{5mm}\forall (a,b)\in{\cal X}_{s}\times{\cal Y}_{d}.
\end{align}
where the conditional probability $p(b|a)$ is defined with respect to the channel $({\cal X}_{s}, p(y_{d}|x_{s}), {\cal Y}_{d})$.  It is convenient to factor out $N(a|{\bf x}^{(n)}_{s})$ from the expression in (\ref{def:strongtypicalityv1}).  This gives:
\begin{align}
    \label{def:strongtypicalityv2}
    T^{(n)}_{\epsilon_{X}}(Y_{d}|{\bf x}^{(n)}_{s}):=\left\{{\bf y}^{(n)}_{d}:\frac{N(a|{\bf x}^{(n)}_{s})}{n}\left|\frac{N((a,b)|({\bf x}^{(n)}_{s},{\bf y}^{(n)}_{d}))}{N(a|{\bf x}^{(n)}_{s})}-p(b|a)\right|<\epsilon_{X}\right\}.
\end{align}

The expression in (\ref{def:strongtypicalityv2}) reveals two factors that affect whether or not a sequence ${\bf y}^{(n)}_{d}$ is strong-conditionally typical.  The first is the frequency of the anchor (transmitted) symbol $a\in{\cal X}_{s}$ defined by the term: 
\begin{align}
\frac{N(a|{\bf x}^{(n)}_{s})}{n}.
\end{align}
The second is the relative frequency of the anchor $a$ and the generated (received) symbol $b\in{\cal Y}_{d}$ with respect to $a$, denoted by: 
\begin{align}
\frac{N((a,b)|({\bf x}^{(n)}_{s},{\bf y}^{(n)}_{d}))}{N(a|{\bf x}^{(n)}_{s})}.
\end{align}
The first is determined by the source sequence ${\bf x}^{(n)}_{s}$, whereas the second is determined (in a probabilistic sense) by the channel $p(y_{d}|x_{s})$.  Our next task is making this statement more precise.  Given the channel $p(y_{d}|x_{s})$ and a specific source codeword ${\bf x}^{(n)}_{s}$, we define the following conditional probability of a received sequence ${\bf y}^{(n)}_{d}$ with respect to a source codeword ${\bf x}^{(n)}_{s}$:

\begin{align}
\label{def:probabilitymassv1}
\mu^{(n)}_{Y_{d}|{\bf x}^{(n)}_{s}}({\bf y}^{(n)}_{d})&:=\prod^{n}_{i=1}p(y^{(i)}_{d}|x^{(i)}_{s}).
\end{align}
Let ${\cal Y}^{(n)}_{d}$ denote the set of all received sequences.  We can define the mass of a subset $A^{(n)}\subseteq{\cal Y}^{(n)}_{d}$ of these sequences using (\ref{def:probabilitymassv1}):
\begin{align}
\label{def:probabilitymassv2}
\mu^{(n)}_{Y_{d}|{\bf x}^{(n)}_{s}}(A^{(n)}_{r})&:=\sum_{{\bf y}^{(n)}_{d}\in A^{(n)}}\mu_{Y_{d}|{\bf x}^{(n)}_{s}}({\bf y}^{(n)}_{d}).
\end{align}
For any source codeword ${\bf x}^{(n)}_{s}$, fix a symbol $a\in{\cal X}_{s}$.  We define two cases.

{\bf Case 1:} \begin{align}
\label{def:case1}
\frac{N(a|{\bf x}^{(n)}_{s})}{n}<\frac{1}{\sqrt{n}},
\end{align}

{\bf Case 2:}\begin{align}
\label{def:case2}
\frac{N(a|{\bf x}^{(n)}_{s})}{n}\geq\frac{1}{\sqrt{n}}.
\end{align}
The trivial case (the rare symbol case) is case 1.  It is straightforward to see that if (\ref{def:case1}) is satisfied then:
\begin{align}
    \label{case1eq1}
    \frac{N(a|{\bf x}^{(n)}_{s})}{n}\left|\frac{N((a,b)|({\bf x}^{(n)}_{s},{\bf y}^{(n)}_{d}))}{N(a|{\bf x}^{(n)}_{s})}-p(b|a)\right|<\frac{1}{\sqrt{n}},
\end{align}
since the following is always true:
\begin{align}
\label{case1eq2}
\left|\frac{N((a,b)|({\bf x}^{(n)}_{s},{\bf y}^{(n)}_{d}))}{N(a|{\bf x}^{(n)}_{s})}-p(b|a)\right|\leq1.
\end{align}
Informally, when an anchor symbol $a\in{\cal X}_{s}$ appears rarely in ${\bf x}^{(n)}_{s}$, the symbol pairs $(y^{(i)}_{d},a)$ generated by $a$ over the channel $p(y_{d}|x_{s})$ are also rare.  These symbol pairs will not generate a useful signature to identify the codeword ${\bf x}^{(n)}_{s}$.  That's no matter, because if $a\in{\cal X}_{s}$ is rare, then some other anchor symbol $a^{\prime}\in{\cal X}_{s}$ is not rare its symbol pairs can generate a useful signature.  

To analyze Case 2, it will be helpful to introduce the following notation of a source sequence ${\bf x}^{(n)}_{s}$ sampled at particular symbols.  For a fixed $a\in{\cal X}_{s}$ and codeword ${\bf x}^{(n)}_{s}$ define:

\begin{align}
\label{def:asampledcodeword}
{\bf x}^{(n)}_{s}[a]=\{(i,x^{(i)}_{s}):x^{(i)}_{s}=a, i\in\{1,\ldots,n\}\}.
\end{align}
Let $|{\bf x}^{(n)}_{s}[a]|$ denote the number of elements in ${\bf x}^{(n)}_{s}[a]$. Observe that:
\begin{align}
    \label{sizeofasampledcodeword}
    |{\bf x}^{(n)}_{s}[a]|=N(a|{\bf x}^{(n)}_{s}).
\end{align}
We can also rebase the measure in (\ref{def:probabilitymassv1}) to sampled codewords.  We define the conditional probability of a destination sequence ${\bf y}^{(n)}_{d}$ with respect to a sampled codeword ${\bf x}^{(n)}_{s}[a]$ as follows:

\begin{align}
    \label{def:probabilitymassv3}
    \mu^{(n)}_{Y_{d}|{\bf x}^{(n)}_{s}[a]}({\bf y}^{(n)}_{d}):=\prod_{i\in\{j:x^{(j)}_{s}=a\}}p(y^{(i)}_{d}|x^{(i)}_{s}).
\end{align}
It will also be convenient to extend our definition of a conditional typical set to sampled source sequence ${\bf x}^{(n)}_{s}[a]$.  For a fixed $\epsilon_{X}>0$, define: 
\begin{align}
    \label{def:strongtypicalityv3}
    T^{(n)}_{\epsilon_{X}}(Y_{d}|{\bf x}^{(n)}_{s}[a]):=\left\{{\bf y}^{(n)}_{d}:\frac{N(a|{\bf x}^{(n)}_{s})}{n}\left|\frac{N((a,b)|({\bf x}^{(n)}_{s},{\bf y}^{(n)}_{d}))}{N(a|{\bf x}^{(n)}_{s})}-p(b|a)\right|<\epsilon_{X}\right\},\hspace{5mm}\forall b\in{\cal Y}_{d}.
\end{align}
The difference between (\ref{def:strongtypicalityv2}) and (\ref{def:strongtypicalityv3}) is that (\ref{def:strongtypicalityv2}) must hold for all $(a,b)\in{\cal X}_{s}\times{\cal Y}_{d}$ whereas (\ref{def:strongtypicalityv3}) must only hold for fixed $a\in{\cal X}_{s}$ and all $b\in{\cal Y}_{d}$.  We now introduce a theorem about the measure of conditionally-typical sets.  
\begin{lemma}
\label{lemma:typicalsets}
Fix $a\in{\cal X}_{s}$.  Given a sequence of codewords $\{{\bf x}^{(n)}_{s}:n\in\mathbb{N}\}$, let $\{{\bf x}^{(n)}_{s}[a]:n\in\mathbb{N}\}$ denote a corresponding sequence of ``$a$-sampled'' codewords as defined in (\ref{def:asampledcodeword}).  Suppose $\{{\bf x}^{(n)}_{s}[a]:n\in\mathbb{N}\}$ satisfies the following condition:
\begin{align}
    \label{cond:sizeofasampledcodeword}
    |{\bf x}^{(n)}_{s}[a]|\geq\sqrt{n}.
\end{align}
For every $\epsilon_{X}>0$, there exists a vanishing sequence $\{\delta^{(n)}_{\epsilon_{X},a}:n\in\mathbb{N}\}$ such that:
\begin{align}
    \label{def:measureoftypicalset}
    \mu_{Y_{d}|{\bf x}^{(n)}_{s}[a]}(T^{(n)}_{\epsilon_{X}}(Y_{d}|{\bf x}^{(n)}_{s}[a]))>1-\delta^{(n)}_{\epsilon_{X},a}.
\end{align}
\end{lemma}
{\bf Proof:}  To prove (\ref{def:measureoftypicalset}) we first make the definition of a conditional typical set in (\ref{def:strongtypicalityv3}) less forgiving by dropping the term: 
\begin{align}
\frac{N(a|{\bf x}^{(n)}_{s})}{n}.
\end{align}
We define this provisional tighter set of sequences as follows:
\begin{align}
    \label{def:provisionaltypicalset}
    \tilde{T}^{(n)}_{\epsilon_{X}}(Y_{d}|{\bf x}^{(n)}_{s}[a]):=\left\{{\bf y}^{(n)}_{d}:\left|\frac{N((a,b)|({\bf x}^{(n)}_{s},{\bf y}^{(n)}_{d}))}{N(a|{\bf x}^{(n)}_{s})}-p(b|a)\right|<\epsilon_{X}\right\},\hspace{5mm}\forall b\in{\cal Y}_{d}.
\end{align}
First, suppose for all $n$: 
\begin{align}
\label{assumption:sqrt}
|{\bf x}^{(n)}_{s}[a]|=\sqrt{n}    
\end{align}
Then under assumption (\ref{assumption:sqrt}), the weak law of large numbers implies there is an associated vanishing $\{\delta^{(n)}_{\epsilon_{X},a}:n\in\mathbb{N}\}$ such that:
\begin{align}
    \label{def:measureofprovisionaltypicalset}
    \mu_{Y_{d}|{\bf x}^{(n)}_{s}[a]}(\tilde{T}^{(n)}_{\epsilon}(Y_{d}|{\bf x}^{(n)}_{s}[a]))>1-\delta^{(n)}_{\epsilon_{X},a}.
\end{align}
We can relax assumption (\ref{assumption:sqrt}), because under the actual theorem assumption (\ref{cond:sizeofasampledcodeword}), we have:
\begin{align}
    \label{assumption:actualtheorem}
    \sqrt{n}\leq |{\bf x}^{(n)}_{s}[a]|\leq n.
\end{align}
The implication of (\ref{assumption:actualtheorem}) is that the size of ${\bf x}^{(n)}_{s}[a]$ is bouncing around the interval $[\sqrt{n}, n]$ for every value of $n$.  But since these sizes are always larger than $\sqrt{n}$, the $\{\delta^{(n)}_{\epsilon_{X},a}:n\in\mathbb{N}\}$ in (\ref{def:measureofprovisionaltypicalset}) that works for (\ref{assumption:sqrt}) is also going to work for (\ref{assumption:actualtheorem}).  Finally, to conclude the proof, we note that:
\begin{align}
    \tilde{T}^{(n)}_{\epsilon}(Y_{d}|{\bf x}^{(n)}_{s}[a])\subseteq T^{(n)}_{\epsilon}(Y_{d}|{\bf x}^{(n)}_{s}[a]), 
\end{align}
since,
\begin{align}
    \frac{N(a|{\bf x}^{(n)}_{s})}{n} \leq 1.
\end{align}
\ensuremath{\blacksquare}

We would like to generalize Lemma \ref{lemma:typicalsets} in two ways.  First, we would like to remove the case 2 condition (\ref{cond:sizeofasampledcodeword}) and second, we would like to use measures and typical sets based on codewords ${\bf x}^{(n)}_{s}$ instead of ``$a$-sampled'' codewords ${\bf x}^{(n)}_{s}[a]$.  Before doing so, it is helpful to see that the definitions in (\ref{def:probabilitymassv1}) and (\ref{def:probabilitymassv3}) imply the following:
\begin{align}
\label{property:decomposemeasure}
\mu_{Y_{d}|{\bf x}^{(n)}_{s}}({\bf y}^{(n)}_{d}):=\prod_{a\in{\cal X}_{s}}\mu^{(n)}_{Y_{d}|{\bf x}^{(n)}_{s}[a]}({\bf y}^{(n)}_{d}).    
\end{align}

\begin{theorem}
    \label{theorem:measureoftypicalset}
    Given a sequence of codewords $\{{\bf x}^{(n)}_{s}:n\in\mathbb{N}\}$ and some fixed $0<\epsilon_{X}<1$, there exists some $N_{\epsilon_{X}}$ and a vanishing sequence $\{\delta^{(n)}_{\epsilon_{X}}:n\in\mathbb{N}\}$ such that for all $n\geq N_{\epsilon_{X}}$:
    \begin{align}
        \label{theorem:measureoftypicalset:eq5}
        \mu_{Y_{d}|{\bf x}^{(n)}_{s}}(T^{(n)}_{\epsilon_{X}}(Y_{d}|{\bf x}^{(n)}_{s}))>1-\delta^{(n)}_{\epsilon_{X}}.
    \end{align}
\end{theorem}
{\bf Proof:}
In view of (\ref{property:decomposemeasure}), we first extend Lemma \ref{lemma:typicalsets} by removing the case 2 condition (\ref{cond:sizeofasampledcodeword}).  For any $\epsilon_{X}>0$ the following is true for all $n\geq (\frac{1}{\epsilon_{X}})^{2}$:
\begin{align}
    \label{property:largeenoughn}
    \frac{1}{\sqrt{n}}<\epsilon_{X}.
\end{align}
Now we make the observation that the definitions (\ref{def:probabilitymassv3}) and (\ref{def:strongtypicalityv3}) imply:
\begin{align}
    \label{theorem:measureoftypicalset:eq2}
    \mu_{Y_{d}|{\bf x}^{(n)}_{s}[a]}(T^{(n)}_{\epsilon_{X}}(Y_{d}|{\bf x}^{(n)}_{s}))&=\mu_{Y_{d}|{\bf x}^{(n)}_{s}[a]}(T^{(n)}_{\epsilon_{X}}(Y_{d}|{\bf x}^{(n)}_{s}[a])).
\end{align}
From Lemma \ref{lemma:typicalsets} and (\ref{property:largeenoughn}), it follows that for all $n\geq(\frac{1}{\epsilon_{X}})^{2}$, we have:
\begin{align}
    \label{theorem:measureoftypicalset:eq1}
    \mu_{Y_{d}|{\bf x}^{(n)}_{s}[a]}(T^{(n)}_{\epsilon_{X}}(Y_{d}|{\bf x}^{(n)}_{s}))>1-\delta^{(n)}_{\epsilon_{X},a}.
\end{align}
Next, we will define:
\begin{align}
    \label{theorem:measureoftypicalset:eq3}
    \delta^{(n)}_{\epsilon_{X}}=|{\cal X}_{s}|\max_{a\in{\cal X}_{s}}\delta^{(n)}_{\epsilon,a}.
\end{align}
It follows from (\ref{property:decomposemeasure}), (\ref{theorem:measureoftypicalset:eq2}) and (\ref{theorem:measureoftypicalset:eq1}) that for all $n\geq(\frac{1}{{\epsilon_{X}}})^{2}$:
\begin{align}
\label{theorem:measureoftypicalset:eq4}
\mu_{Y_{d}|{\bf x}^{(n)}_{s}}(T^{(n)}_{\epsilon_{X}}(Y_{d}|{\bf x}^{(n)}_{s}))>1-\delta^{(n)}_{\epsilon_{X}}.
\end{align}
The bound in (\ref{theorem:measureoftypicalset:eq4}) comes from (\ref{property:decomposemeasure}), (\ref{theorem:measureoftypicalset:eq1}), (\ref{theorem:measureoftypicalset:eq3}) and the inequality $(1-c)^{|{\cal X}_{s}|}>1-|{\cal X}_{s}|c$, where $c=\max_{a}\delta^{(n)}_{{\epsilon_{X}},a}$. \ensuremath{\blacksquare}

We define the type of a codeword as a point on the simplex over ${\cal X}_{s}$:
\begin{align}
{\bf p}^{(n)}_{s}:=\left\{\left(a, \frac{N(a|{\bf x}^{(n)}_{s})}{n}\right):a\in{\cal X}_{s}\right\}.
\end{align}
The salient point about Theorem \ref{theorem:measureoftypicalset} is that there are no assumptions on $\{{\bf x}^{(n)}_{s}:n\in\mathbb{N}\}$.  We do not assume that ${\bf p}^{(n)}_{s}$ converges to any point on the simplex.  The vanishing $\{\delta^{(n)}_{\epsilon}:n\in\mathbb{N}\}$ in  holds uniformly for all sequences $\{{\bf x}^{(n)}_{s}:n\in\mathbb{N}\}$ regardless of the corresponding types $\{{\bf p}^{(n)}_{s}:n\in\mathbb{N}\}$.  It will be convenient to replace the $\epsilon_{X}$ in $T^{(n)}_{\epsilon_{X}}(Y_{d}|{\bf x}^{(n)}_{s})$ with a  vanishing $\{\epsilon^{(n)}_{X}:n\in\mathbb{N}\}$.  We have the following corollary:
\begin{corollary}
\label{corollary:measureofsharptypicalsets}
Given a sequence of codewords $\{{\bf x}^{(n)}_{s}:n\in\mathbb{N}\}$, there exists some $N_{0}$, some vanishing $\{\epsilon_{n}:n\in\mathbb{N}\}$ and a vanishing $\{\delta^{(n)}_{X}:n\in\mathbb{N}\}$ such that for all $n\geq N_{0}$:
    \begin{align}
        \label{corollary:measureofsharptypicalsets:eq1}
        \mu_{Y_{d}|{\bf x}^{(n)}_{s}}(T^{(n)}_{\epsilon_{n}}(Y_{d}|{\bf x}^{(n)}_{s}))>1-\delta^{(n)}_{X}.
    \end{align}
\end{corollary}
{\bf Proof:}  From Theorem \ref{theorem:measureoftypicalset}, every fixed $\epsilon$ has a corresponding sequence $\{\delta^{(n)}_{\epsilon}:n\in\mathbb{N}\}$ that satisfies (\ref{theorem:measureoftypicalset:eq5}).  It follows that we can find some vanishing $\{\epsilon_{n}:n\in\mathbb{N}\}$ such that $\{\delta^{(n)}_{\epsilon_{n}}:n\in\mathbb{N}\}$ is also vanishing and $\delta^{(n)}_{X}=\delta^{(n)}_{\epsilon_{n}}$ satisfies (\ref{corollary:measureofsharptypicalsets:eq1}).  We just need to make sure that this $\{\epsilon_{n}:n\in\mathbb{N}\}$ is vanishing slowly enough to ensure $n\geq (\frac{1}{\epsilon_{n}})^{2}$. \ensuremath{\blacksquare}

The distinction between a vanishing epsilon $\{\epsilon_{n}:n\in\mathbb{N}\}$ and a fixed epsilon $\epsilon$ plays an important role in the argument for the capacity of the relay channel (and also the strong converse for the point-to-point channel).  It merits a special definition of typical sets - sharp typicality.  For the $\{\epsilon_{n}:n\in\mathbb{N}\}$ in Corollary \ref{corollary:measureofsharptypicalsets}, we define the set of sharply strong-typical ${\bf y}_{d}$ sequences as follows:
\begin{align}
    \label{def:sharplyconditionaltypicalsequences}
    T^{(n)}_{*}(Y_{d}|{\bf x}^{(n)}_{s}):=T^{(n)}_{\epsilon_{n}}(Y_{d}|{\bf x}^{(n)}_{s}).
\end{align}

We need the notion of sharp typicality to tighten the bounds on the normalized exponent on the volume of strong typical sets.  However, there is a subtle issue that must be addressed.  Our approach will use the fact that strongly typical sequences are also weakly typical, and weakly typical sequences have known bounds on their   normalized exponents.  The problem is that weak-conditionally typical sequences depend on the empirical symbol distribution ${\bf p}^{(n)}_{s}$ of the anchor sequence ${\bf x}^{(n)}_{s}$ and we have deliberately refrained from imposing conditions on this distribution.  In particular, we do not assume that ${\bf p}^{(n)}_{s}$ converges to anything.  Instead, we will label the conditional entropy induced by the type ${\bf p}^{(n)}_{s}$ and the channel probability transition matrix $p(y_{s}|x_{s})$ and define the volume in terms of that entropy.  Given a sequence of codewords $\{{\bf x}^{(n)}_{s}:n\in\mathbb{N}\}$ and corresponding types $\{{\bf p}^{(n)}_{s}:n\in\mathbb{N}\}$ define the distribution:
\begin{align}
    \label{def:typeinducedmarginal}
    p^{(n)}_{X}(x_{s}):={\bf p}^{(n)}_{s}.
\end{align}
We define the joint distribution induced by $p^{(n)}_{X}(x_{s})$ and $p(y_{d}|x_{s})$:
\begin{align}
    \label{def:typeinducedjoint}
    p^{(n)}(x_{s},y_{d}):=p(y_{d}|x_{s})p^{(n)}_{X}(x_{s}).
\end{align}
We have a second corollary of Theorem \ref{theorem:measureoftypicalset} that gives bounds on the normalized exponent of the volume of sharp conditionally-typical sets.
\begin{corollary}
    \label{corollary:volumeofsharplytypicalsets}
    Given a sequence of codewords $\{{\bf x}^{(n)}_{s}:n\in\mathbb{N}\}$ with corresponding types $\{{\bf p}^{(n)}_{s}:n\in\mathbb{N}\}$, let $H^{(n)}(Y_{d}|X_{s})$ denote the conditional entropy induced by the joint distribution $p^{(n)}(x_{s},y_{d})$ defined in (\ref{def:typeinducedjoint}).  There exists some vanishing $\{\alpha^{(n)}_{X}:n\in\mathbb{N}\}$ such that:
    \begin{align}
        \label{corollary:volumeofsharplytypicalsets:eq1}
        H^{(n)}(Y_{d}|X_{s})-\alpha^{(n)}_{X}<\frac{1}{n}\log_{2}|T^{(n)}_{*}(Y_{d}|{\bf x}^{(n)}_{s})|<H^{(n)}(Y_{d}|X_{s})+\alpha^{(n)}_{X}.
    \end{align}
\end{corollary}
{\bf Proof:} For any ${\bf y}^{(n)}_{d}\in T^{(n)}_{*}(Y_{d}|{\bf x}^{(n)}_{s})$, there is some constant $K$ with respect to $n$ that depends only on $|{\cal X}_{s}|$ and $|{\cal Y}_{d}|$, and some vanishing $\{\epsilon_{n}:n\in\mathbb{N}\}$ such that:
\begin{align}
    \label{massofasharptypicalsequence}
    2^{-n(H^{(n)}(Y_{d}|X_{s})+K\epsilon_{n})}<\mu_{Y_{d}|{\bf x}^{(n)}_{s}}({\bf y}^{(n)}_{d})<2^{-n(H^{(n)}(Y_{d}|X_{s})-K\epsilon_{n})}.
\end{align}
The bounds in (\ref{massofasharptypicalsequence}) follow because strong typicality with respect to $\epsilon_{n}$ implies weak typicality with respect to $K\epsilon_{n}$, where $K$ is a constant with respect to $n$ that depends only on $|{\cal X}_{s}|$ and $|{\cal Y}_{d}|$.  Now Corollary \ref{corollary:measureofsharptypicalsets} implies that there is a vanishing $\{\delta^{(n)}_{X}:n\in\mathbb{N}\}$ such that:
\begin{align}
    \label{corollary:volumeofsharplytypicalsets:eq2}
    \mu_{Y_{d}|{\bf x}^{(n)}_{s}}(T^{(n)}_{*}(Y_{d}|{\bf x}^{(n)}_{s}))>1-\delta^{(n)}_{X}.
\end{align}
Define:
\begin{align}
    \label{corollary:measureofsharplytypicalsets:eq3}
    \alpha^{(n)}_{X}:=K\epsilon_{n}+\frac{1}{n}\log_{2}(1-\delta^{(n)}_{X}).
\end{align}
Then (\ref{corollary:volumeofsharplytypicalsets:eq1}) follows from (\ref{massofasharptypicalsequence}), (\ref{corollary:volumeofsharplytypicalsets:eq2}) and (\ref{corollary:measureofsharplytypicalsets:eq3}).
\ensuremath{\blacksquare}

The bounds on the normalized exponent of the volume in (\ref{corollary:volumeofsharplytypicalsets:eq1}) apply uniformly to any sequence $\{{\bf x}^{(n)}_{s}:n\in\mathbb{N}\}$ without restriction on the corresponding types $\{{\bf p}^{(n)}:n\in\mathbb{N}\}$.  We will use these bounds to develop sphere-packing arguments for the point-to-point discrete-memoryless channel and the relay channel.  

We have defined conditionally-typical sequences, now let us define marginally (or ordinary) typical ${\bf y}^{(n)}_{d}$-sequences in the ${\cal Y}^{(n)}_{d}$ space.  For a fixed $\epsilon_{Y}>0$ define:
\begin{align}
    \label{def:typicalreceivedsequences}
    T^{(n)}_{\epsilon_{Y}}(Y_{d}):=\left\{{\bf y}^{(n)}_{d}:\frac{N(a|{\bf y}^{(n)}_{d})}{n}<\epsilon_{Y}\right\},\hspace{5mm}\forall a\in{\cal Y}_{d}.
\end{align}

We continue to refer to the sequence of codewords $\{{\bf x}^{(n)}_{s}:n\in\mathbb{N}\}$ and its associated types $\{{\bf p}^{(n)}_{s}:n\in\mathbb{N}\}$.  Given the marginal $p^{(n)}_{X}(\cdot)$ in (\ref{def:typeinducedmarginal}), define:
\begin{align}
\label{marginalYinducedbytype}
p^{(n)}_{Y}(\cdot):=\sum_{x_{s\in{\cal X}_{s}}}p(\cdot|x_{s})p^{(n)}_{X}(x_{s})
\end{align}
For the $n$-length sequence ${\bf y}^{(n)}_{d}=\{(i,y^{(i)}_{d}):y^{(i)}_{d}\in{\cal Y}_{d},1\leq i\leq n\}$, define the measure:
\begin{align}
\label{measureofY}
\mu^{(n)}_{Y_{d}}({\bf y}^{(n)}_{d}):=\prod^{n}_{i=1}p^{(n)}_{Y}(y^{(i)}_{d}),
\end{align}
where the measure $\mu^{(n)}_{Y_{d}}(\cdot)$ is indexed to the marginal $p^{(n)}_{Y}(\cdot)$ induced by the types $\{{\bf p}^{(n)}_{s}:n\in\mathbb{N}\}$ and the channel distribution $p(y_{d}|x_{s})$ via (\ref{def:typeinducedmarginal}), (\ref{marginalYinducedbytype}) and (\ref{measureofY}).  We have the following lemma:
\begin{theorem}
\label{theorem:measureofamarginaltypicalset}
Given the sequence of codewords $\{{\bf x}^{(n)}_{s}:n\in\mathbb{N}\}$ with types $\{{\bf p}^{(n)}_{s}:n\in\mathbb{N}\}$ and any fixed $\epsilon_{Y}>0$, there exists a vanishing sequence $\{\delta^{(n)}_{\epsilon_{Y}}:n\in\mathbb{N}\}$ such that:
\begin{align}
    \label{theorem:measureofamarginaltypicalset:eq1}
    \mu^{(n)}_{Y_{d}}(T^{(n)}_{\epsilon_{Y}}(Y_{d}))>1-\delta^{(n)}_{\epsilon_{Y}},
\end{align}
where $\mu^{(n)}_{Y}(\cdot)$ is defined by (\ref{measureofY}).  The vanishing sequence $\{\delta^{(n)}_{\epsilon_{Y}}:n\in\mathbb{N}\}$ applies uniformly to all possible sequences of codewords $\{{\bf x}^{(n)}_{s}:n\in\mathbb{N}\}$.
\end{theorem}
{\bf Proof:} The proof follows from the weak law of large numbers and the fact that ${\cal Y}_{d}$ is discrete and finite.
\begin{corollary}
    \label{corollary:measureofmarginalYtypicalset}
    Given a sequence of codewords $\{{\bf x}^{(n)}_{s}:n\in\mathbb{N}\}$ with corresponding types $\{{\bf p}^{(n)}_{s}:n\in\mathbb{N}\}$ and the induced distribution (\ref{measureofY}), there exists some vanishing $\{\epsilon_{n}:n\in\mathbb{N}\}$ and a vanishing $\{\delta^{(n)}_{Y}:n\in\mathbb{N}\}$ such that:
    \begin{align}
        \label{corollary:measureofmarginalYtypicalset:eq1}
        \mu^{(n)}_{Y_{d}}(T^{(n)}_{\epsilon_{n}}(Y_{d}))>1-\delta^{(n)}_{Y}.
    \end{align}
\end{corollary}
{\bf Proof:} From Theorem \ref{theorem:measureofamarginaltypicalset}, every fixed $\epsilon$ has a corresponding sequence $\{\delta^{(n)}_{\epsilon}:n\in\mathbb{N}\}$ that satisfies (\ref{theorem:measureofamarginaltypicalset:eq1}).  It follows that we can find some vanishing $\{\epsilon_{n}:n\in\mathbb{N}\}$ such that $\{\delta^{(n)}_{\epsilon_{n}}:n\in\mathbb{N}\}$ is also vanishing and $\delta^{(n)}_{Y}:=\delta^{(n)}_{\epsilon_{n}}$ satisfies (\ref{corollary:measureofmarginalYtypicalset:eq1}).
\ensuremath{\blacksquare}

For the $\{\epsilon_{n}:n\in\mathbb{N}\}$ in Corollary \ref{corollary:measureofmarginalYtypicalset}, we define the set of sharply strong-typical ${\bf y}_{d}$-sequences as follows:
\begin{align}
    T^{(n)}_{*}(Y_{d}):=T^{(n)}_{\epsilon_{n}}(Y_{d}).
\end{align}
We have another corollary of Theorem \ref{theorem:measureofamarginaltypicalset} that gives bounds on the normalized exponent of the volume of sharp typical sets.
\begin{corollary}
    \label{corollary:volumeofmarginalYinducedbytype}
    Given a sequence of codewords $\{{\bf x}^{(n)}_{s}:n\in\mathbb{N}\}$ with corresponding types $\{{\bf p}^{(n)}_{s}:n\in\mathbb{N}\}$, let $H^{(n)}(Y_{d})$ denote the entropy induced by the marginal distribution $p^{(n)}_{Y}(\cdot)$ defined in (\ref{marginalYinducedbytype}).  There exists some vanishing $\{\alpha^{(n)}_{Y}:n\in\mathbb{N}\}$ such that:
    \begin{align}
        \label{corollary:volumeofmarginalYinducedbytype:eq4}
        H^{(n)}(Y_{d})-\alpha^{(n)}_{Y}<\frac{1}{n}\log_{2}|T^{(n)}_{*}(Y_{d})|<H^{(n)}(Y_{d})+\alpha^{(n)}_{Y}.
    \end{align}
\end{corollary} 
{\bf Proof:} For any ${\bf y}^{(n)}_{d}\in T^{(n)}_{*}(Y_{d})$, there is some constant $K$ with respect to $n$ that depends only on $|{\cal Y}_{d}|$, and some vanishing $\{\epsilon_{n}:n\in\mathbb{N}\}$ such that:
\begin{align}
    \label{corollary:volumeofmarginalYinducedbytype:eq1}
    2^{-n(H^{(n)}(Y_{d})+k\epsilon_{n})}<\mu^{(n)}_{{Y}_{d}}({\bf y}^{(n)}_{d})<2^{-n(H^{(n)}(Y_{d})-k\epsilon_{n})}.
\end{align}
The bounds in (\ref{corollary:volumeofmarginalYinducedbytype:eq1}) follow because strong typicality with respect to $\epsilon_{n}$ implies weak typicality with respect to $K\epsilon_{n}$, where $K$ is a constant with respect to $n$ that depends only on $|{\cal Y}_{d}|$.  Now Corollary \ref{corollary:measureofmarginalYtypicalset} implies that there is a vanishing $\{\delta^{(n)}_{Y}:n\in\mathbb{N}\}$ such that:
\begin{align}
    \label{corollary:volumeofmarginalYinducedbytype:eq2}
    \mu^{(n)}_{Y_{d}}(T^{(n)}_{*}(Y_{d}))>1-\delta^{(n)}_{Y}.
\end{align}
Define:
\begin{align}
    \label{corollary:volumeofmarginalYinducedbytype:eq3}
    \alpha^{(n)}_{Y}:=K\epsilon_{n}+\frac{1}{n}\log_{2}(1-\delta^{(n)}_{Y}).
\end{align}
Then (\ref{corollary:volumeofmarginalYinducedbytype:eq4}) follows from (\ref{corollary:volumeofmarginalYinducedbytype:eq1}), (\ref{corollary:volumeofmarginalYinducedbytype:eq2}), and (\ref{corollary:volumeofmarginalYinducedbytype:eq3}).

\section{The Point-to-Point Discrete-Memoryless Channel}
\label{sec:pointtopointchannel}
The point-to-point channel is defined by the triple $({\cal X}_{s},p(y_{d}|x_{s}),{\cal Y}_{d})$, where ${\cal X}_{s}$ denotes the set of symbols sent by the source and ${\cal Y}_{d}$ denotes the set of symbols received by the destination.  The channel statistics are defined by the probability transition matrix $p(y_{d}|x_{s})$, where $y_{d}\in{\cal Y}_{d}$ and $x_{s}\in{\cal X}_{s}$.  Associated with the channel are sequences of message sets, source encoding functions, and destination decoding functions all indexed by $n$.  Let $M^{(n)}_{s}:=\{1,\ldots, |M^{(n)}_{s}|\}$ denote the source message set.  Let $f^{(n)}_{s}: M^{(n)}_{s}\rightarrow{\bf x}^{(n)}_{s}$ denote the source encoding function, where ${\bf x}^{(n)}_{s}:=\{x^{(i)}_{s}\in{\cal X}_{s}:1\leq i\leq n\}$.  Let $g^{(n)}_{d}:{\bf y}^{(n)}_{d}\rightarrow\hat{M}^{(n)}_{s}$ denote the destination decoding function, where ${\bf y}^{(n)}_{d}:=\{y^{(i)}_{d}\in{\cal Y}_{d}:1\leq i\leq n\}$ and $\hat{M}^{(n)}_{d}:=\{1,\ldots, |M^{(n)}_{s}|\}$.  Let ${\cal Y}^{(n)}_{d}$ denote the set of $n$-length ${\bf y}_{d}$ sequences.

The estimated message set at the destination $\hat{M}^{(n)}_{d}$ has the same size as the source message set $M^{(n)}_{s}$.  The destination decoding function $g^{(n)}_{d}$ creates a decoding set $A^{(n)}_{m}\in{\cal Y}^{(n)}_{d}$ associated with every message $m\in\hat{M}^{(n)}_{d}$.  We will assume the decoding sets are disjoint.  If message $m\in M^{(n)}_{s}$ is selected by the source, the source encoding function $f^{(n)}_{s}$ sends the sequence ${\bf x}^{(n)}_{s}(m)$.  This sequence passes through the channel and generates the random received sequence ${\bf Y}^{(n)}_{d}$ at the destination.  The probability of the destination correctly decoding the message is:
\begin{align}
\label{pointtopoint:decodingsetprob}
\mu_{Y_{d}|{\bf x}^{(n)}_{s}(m)}(A^{(n)}_{m}).
\end{align}

The measure $\mu_{Y_{d}|{\bf x}^{(n)}_{s}(m)}(\cdot)$ in (\ref{pointtopoint:decodingsetprob}) is defined in (\ref{def:probabilitymassv1}).  The capacity $C$ of the point-to-point channel is the supremum of all rates $R\doteq|M^{(n)}_{s}|$, with a corresponding sequence of encoding/decoding functions $\{f^{(n)}_{s}, g^{(n)}_{d}:n\in\mathbb{N}\}$ that satisfies:
\begin{align}
    \label{pointtopoint:vanishingdecodingerror}
    \lim_{n\rightarrow\infty}\min_{m\in M^{(n)}_{s}}\mu_{Y_{d}|{\bf x}^{(n)}_{s}(m)}(A^{(n)}_{m})=1.
\end{align}
For the record, $R\doteq|M^{(n)}_{s}|$ means $R=\lim_{n\rightarrow\infty}\frac{1}{n}\log_{2}|M^{(n)}_{s}|$.  The expression in (\ref{pointtopoint:vanishingdecodingerror}) implies that the worst case probability of decoding error vanishes.
\begin{theorem}
    \label{theorem:capacityofthepointtopointchannel}
    The capacity of the point-to-point discrete-memoryless channel satisfies:
    \begin{align}
    \label{theorem:capacityofthepointtopointchannel:eq1}
        C\leq\max_{p(x_{s})}I(X_{s};Y_{d}).
    \end{align}
\end{theorem}
{\bf Proof:}  The source encoding function induces a source codebook of codewords ${\cal C}^{(n)}_{s}(R):=\{{\bf x}^{(n)}_{s}(m):m\in M^{(n)}_{s}\}$.  Each codeword ${\bf x}^{(n)}_{s}(m)$ has a corresponding type ${\bf p}^{(n)}_{s}(m)$.  We would like to restrict our attention to significant types, that is, types which have enough codewords to deserve consideration.  First, we define a type counting function.  Let $N({\cal C}^{(n)}_{s}(R)|{\bf p}^{(n)}_{s})$ denote the number of codewords in ${\cal C}^{(n)}_{s}(R)$ that have the type ${\bf p}^{(n)}_{s}$.  Next, we name a vanishing quantity implicit in the problem formulation.  Since $R\doteq|M^{(n)}_{s}|$, we know there is some vanishing $\{\epsilon^{(n)}_{1}:n\in\mathbb{N}\}$ such that:
\begin{align}
    \label{boundsonsizeofcodebook}
    R-\epsilon^{(n)}_{1}<\frac{1}{n}\log_{2}|M^{(n)}_{s}|<R+\epsilon^{(n)}_{1}.
\end{align}
Finally, we say that the type ${\bf p}^{(n)}_{s}$ is significant if it satisfies the following condition:
\begin{align}
    \label{definitionofsignificanttypes}
    \frac{1}{n}\log_{2}|N({\cal C}^{(n)}_{s}(R)|{\bf p}^{(n)}_{s})|>R-\epsilon^{(n)}_{1}-\frac{(\log_{2}n)(\log_{2}n)}{n}.
\end{align}
A type is insignificant if it is not significant.  Given a sequence of codebooks $\{{\cal C}^{(n)}(R):n\in\mathbb{N}\}$ induced by the sequence of source encoding functions $\{f^{(n)}_{s}:n\in\mathbb{N}\}$, let $\tilde{{\cal C}}^{(n)}_{s}(R)$ denote the set of codewords from ${\cal C}^{(n)}_{s}(R)$ with insignificant types.  Since all codewords have the same mass, the mass of $\tilde{{\cal C}}^{(n)}_{s}(R)$ is exactly:
\begin{align}
    \label{insignificanttypeprobability}
    \frac{|\tilde{{\cal C}}^{(n)}_{s}(R)|}{|{\cal C}^{(n)}_{s}(R)|}.
\end{align}
Let ${\cal X}^{(n)}_{s}$ denote the set of all possible $n$-length ${\bf x}_{s}$-sequences.  There are at most $n^{K_{T}}$ distinct types in ${\cal X}^{(n)}_{s}$ for some constant $K_{T}$ with respect to $n$.  So we can use (\ref{definitionofsignificanttypes}) to get the following bound on the number of insignificant types in ${\cal C}^{(n)}_{s}(R)$:
\begin{align}
    \label{upperboundonnumberofinsignificanttypes}
    \frac{1}{n}\log_{2}|\tilde{C}^{(n)}_{s}(R)|<R-\epsilon^{(n)}_{1}-\frac{(\log_{2}n)(\log_{2}n)}{n}+\frac{K_{T}\log_{2}n}{n}.
\end{align}
From the definition of ${\cal C}^{(n)}_{s}(R)$ we know that:
\begin{align}
    \label{sizeofcodebookissameassizeofmessageset}
    |C^{(n)}_{s}(R)|=|M^{(n)}_{s}|.
\end{align}
It follows from (\ref{boundsonsizeofcodebook}), (\ref{upperboundonnumberofinsignificanttypes}), and (\ref{sizeofcodebookissameassizeofmessageset}) that:
\begin{align}
    \label{boundonmassofinsignificanttypes}
    \log_{2}\left(\frac{|\tilde{{\cal C}}^{(n)}_{s}(R)|}{|{\cal C}^{(n)}_{s}(R)|}\right)<K_{T}\log_{2}n -(\log_{2}n)(\log_{2}n).
\end{align}
It follows from (\ref{boundonmassofinsignificanttypes}) that:
\begin{align}
    \label{asymptoticmassofinsignificanttypes}
    \lim_{n\rightarrow\infty}\frac{|\tilde{{\cal C}}^{(n)}_{s}(R)|}{|{\cal C}^{(n)}_{s}(R)|}=0.
\end{align}
Since the mass of the insignificant types is vanishing, as per (\ref{asymptoticmassofinsignificanttypes}), we can restrict our attention to significant types.  In fact, we can go even further and restrict our attention to one specific significant type in each codebook - a special subcodebook consisting of codewords sharing exactly one significant type.  If the original sequence of codebooks satisfies (\ref{pointtopoint:vanishingdecodingerror}) then so must the sequence of special subcodebooks, since each special subcodebook has fewer codewords than the original.  Hence, we will assume a codebook sequence $\{{\cal C}^{(n)}_{s}(R):n\in\mathbb{N}\}$ where all the codewords in a codebook share the same type.  We need to show that (\ref{pointtopoint:vanishingdecodingerror}) implies (\ref{theorem:capacityofthepointtopointchannel:eq1}).

We name more vanishing quantities in the formulation.  If (\ref{pointtopoint:vanishingdecodingerror}) holds, then there is a vanishing $\{\epsilon^{(n)}_{2}:n\in\mathbb{N}\}$ such that for every $n\in\mathbb{N}$:
\begin{align}
    \label{epsilononvanishingdecodingerror}
    \min_{m\in M^{(n)}_{s}}\mu_{Y_{d}|{\bf x}^{(n)}_{s}(m)}(A^{(n)}_{m})>1-\epsilon^{(n)}_{2}.
\end{align}
Since each codeword ${\bf x}^{(n)}_{s}\in{\cal C}^{(n)}_{s}(R)$ shares the same type ${\bf p}^{(n)}_{s}$ (as per our previous assumption), every codeword ${\bf x}^{(n)}_{s}(m)$ is a permutation $\Pi^{(n)}_{m}$ of some base codeword which we will choose to be ${\bf x}^{(n)}_{s}(1)$.  We have the following set of equations:
\begin{align}
    \nonumber
    \mu_{Y_{d}|{\bf x}^{(n)}_{s}(m)}(T^{(n)}_{*}(Y_{d}|{\bf x}^{(n)}_{s}(m)))&=\mu_{Y_{d}|\Pi^{(n)}_{m}({\bf x}^{(n)}_{s}(m))}(T^{(n)}_{*}(Y_{d}|\Pi^{(n)}_{m}({\bf x}^{(n)}_{s}(m))),\\
    \nonumber
    &=\mu_{Y_{d}|{\bf x}^{(n)}_{s}(1)}([\Pi^{(n)}_{m}]^{-1}(T^{(n)}_{*}(Y_{d}|\Pi^{(n)}_{m}({\bf x}^{(n)}_{s}(m)))),\\
    \nonumber
    &=\mu_{Y_{d}|{\bf x}^{(n)}_{s}(1)}([\Pi^{(n)}_{m}]^{-1}\Pi^{(n)}_{m}(T^{(n)}_{*}(Y_{d}|{\bf x}^{(n)}_{s}(1))),\\
    \label{permutationprobability}
    &=\mu_{Y_{d}|{\bf x}^{(n)}_{s}(1)}(T^{(n)}_{*}(Y_{d}|{\bf x}^{(n)}_{s}(1))).
\end{align}
From (\ref{permutationprobability}) and Corollary \ref{corollary:measureofsharptypicalsets}, there exists some vanishing sequence $\{\delta^{(n)}_{X}:n\in\mathbb{N}\}$ such that:
\begin{align}
    \label{minprobabilityofatypicalset}
    \min_{m\in M^{(n)}_{s}}\mu_{Y_{d}|{\bf x}^{(n)}_{s}(m)}(T^{(n)}_{*}(Y_{d}|{\bf x}^{(n)}_{s}(m)))>1-\delta^{(n)}_{X}
\end{align}
Now we link (\ref{epsilononvanishingdecodingerror}) with (\ref{minprobabilityofatypicalset}) to get:
\begin{align}
    \label{intersectionprobability}
    \min_{m\in M^{(n)}_{s}}\mu_{Y_{d}|{\bf x}^{(n)}_{s}(m)}(A^{(n)}_{m}\cap T^{(n)}_{*}(Y_{d}|{\bf x}^{(n)}_{s}(m)))>1-\epsilon^{(n)}_{2}-\delta^{(n)}_{X}
\end{align}
Moreover, for every ${\bf y}^{(n)}_{d}\in T^{(n)}_{*}(Y_{d}|{\bf x}^{(n)}_{s}(m))$ and every $m\in M^{(n)}_{s}$, there exists some vanishing $\{\epsilon^{(n)}_{3}:n\in\mathbb{N}\}$ such that:
\begin{align}
    \label{massofatypicalysequence}
    2^{-n(H^{(n)}(Y_{d}|X_{s})+\epsilon^{(n)}_{3})}<\mu_{Y_{d}|{\bf x}^{(n)}_{s}(m)}({\bf y}^{(n)}_{d})<2^{-n(H^{(n)}(Y_{d}|X_{s})-\epsilon^{(n)}_{3})},
\end{align}
where the conditional entropy $H^{(n)}(Y_{d}|X_{s})$ is defined with respect to the joint distribution $p^{(n)}(x_{s},y_{d})=p(y_{d}|x_{s})p^{(n)}_{X}(x_{s})$ and $p^{(n)}_{X}(x_{s})={\bf p}^{(n)}_{s}$.  To justify (\ref{massofatypicalysequence}), observe that a ${\bf y}^{(n)}_{d}$-sequence strongly typical with respect to $\epsilon_{n}$ is weakly typical with respect to $K\epsilon_{n}$ for some constant $K$ (with respect to $n$) that depends only on $|{\cal X}_{s}|$ and $|{\cal Y}_{d}|$.  We get (\ref{massofatypicalysequence}) by setting $\epsilon^{(n)}_{3}=K\epsilon^{(n)}$.  The bound (\ref{massofatypicalysequence}) applies uniformly to all ${\bf x}^{(n)}_{s}\in{\cal C}^{(n)}_{s}(R)$ (or equivalently, all $m\in M^{(n)}_{s}$) because we assume that all these codewords share the same type ${\bf p}^{(n)}_{s}$.  Define:
\begin{align}
    \label{epsilon4}
    \epsilon^{(n)}_{4}:=\frac{1}{n}\log_{2}(1-\epsilon^{(n)}_{2}-\delta^{(n)}_{X})
\end{align}
We use (\ref{intersectionprobability}) and (\ref{massofatypicalysequence}) to get a lower bound on $|A^{(n)}_{m}\cap T^{(n)}_{*}(Y_{d}|{\bf x}^{(n)}(m))|$ for all $m\in M^{(n)}_{s}$:
\begin{align}
    \label{minimumvolumeofdecodingset}
    \min_{m\in M^{(n)}_{s}}|A^{(n)}_{m}\cap T^{(n)}_{*}(Y_{d}|{\bf x}^{(n)}_{s}(m))|>2^{n(H^{(n)}(Y_{d}|X_{s})+\epsilon^{(n)}_{4}-\epsilon^{(n)}_{3})}.
\end{align}
Now comparing the definition of $\epsilon_{X}$-conditionally-typical sequences in (\ref{def:strongtypicalityv1}) and ordinary typical sequences in (\ref{def:typicalreceivedsequences}), gives:
\begin{align}
    \label{typicalreceivedsequencesinthebigset}
    T^{(n)}_{\epsilon_{X}}(Y_{d}|{\bf x}^{(n)}_{s}(m))\subseteq T^{(n)}_{|{\cal X}_{s}|\epsilon_{X}}(Y_{d}),
\end{align}
for all $m\in M^{(n)}_{s}$.  From (\ref{typicalreceivedsequencesinthebigset}), we can choose a vanishing $\{\epsilon^{(n)}_{Y}:n\in\mathbb{N}\}$ to define $T^{(n)}_{*}(Y_{d})$ such that:
\begin{align}
    T^{(n)}_{*}(Y_{d}|{\bf x}^{(n)}_{s}(m))\subseteq T^{(n)}_{*}(Y_{d}),
\end{align}
for all $m\in M^{(n)}_{s}$.  Now Corollary \ref{corollary:volumeofmarginalYinducedbytype} implies there is some vanishing $\{\alpha^{(n)}_{Y}:n\in\mathbb{N}\}$ such that:
\begin{align}
    \label{sizeoftypicalysequences}
    |T^{(n)}_{*}(Y_{d})|<2^{n(H^{(n)}(Y_{d})+\alpha^{(n)}_{Y})},
\end{align}
where the entropy $H^{(n)}(Y_{d})$ is defined with respect to the distribution $p^{(n)}(y_{d})=\sum_{x_{s}\in{\cal X}_{s}}p^{(n)}(y_{d}|x_{s})p^{(n)}_{X}(x_{s})$ and $p^{(n)}_{X}(x_{s}):={\bf p}^{(n)}_{s}$.  Putting together (\ref{boundsonsizeofcodebook}), (\ref{minimumvolumeofdecodingset}), and (\ref{sizeoftypicalysequences}) gives:
\begin{align}
    \label{outerbound}
    R<I^{(n)}(X_{s};Y_{d})+\epsilon^{(n)}_{1}+\alpha^{(n)}_{Y}+\epsilon^{(n)}_{3}-\epsilon^{(n)}_{4},
\end{align}
for all $n\in\mathbb{N}$ and vanishing $\{\epsilon^{(n)}_{1},\alpha^{(n)}_{Y},\epsilon^{(n)}_{3},\epsilon^{(n)}_{4}\}$.  To close out the argument, we address the problem that $p^{(n)}_{X}(x_{s}):={\bf P}^{(n)}_{s}$ isn't converging to anything.  Since the simplex over ${\cal X}_{s}$ is compact, there exists a point $p(x_{s})$ on the simplex and a subsequence $\{n_{k}:k\in\mathbb{N}\}$ such that:
\begin{align}
    \label{inputdistribution}
    \lim_{k\rightarrow\infty}p^{(n_{k})}(x_{s})=p(x_{s}).
\end{align}
Therefore,
\begin{align}
    \label{mutualinformation}
    \lim_{k\rightarrow\infty}\left(I^{(n_{k})}(X_{s};Y_{d})+\epsilon^{(n_{k})}_{1}+\alpha^{(n_{k})}_{Y}+\epsilon^{(n_{k})}_{3}-\epsilon^{(n_{k})}_{4}\right)=I(X_{s};Y_{d}),
\end{align}
where $I(X_{s};Y_{d})$ is defined with respect to the joint distribution $p(x_{s},y_{d})=p(y_{d}|x_{s})p(x_{s})$ and $p(x_{s})$ comes from (\ref{inputdistribution}).  The proof follows from (\ref{outerbound}) and (\ref{mutualinformation}). \ensuremath{\blacksquare}

Two techniques from the proof of Theorem \ref{theorem:capacityofthepointtopointchannel} will be useful in finding the capacity of the relay channel.  First, is our focus on significant types, which allows us to assume that all codewords in a codebook share the same type.  Second, is the use of compactness to get a subsequence where the types converge to a point on the simplex.

\section{The Discrete-Memoryless Relay Channel}
\label{sec:therelaychannel}
The relay channel is defined by the triple $({\cal X}_{s}\times{\cal X}_{r},p(y_{r},y_{d}|x_{s},x_{r}),{\cal Y}_{r}\times{\cal Y}_{d})$, where ${\cal X}_{s}$ and ${\cal X}_{r}$ denote the set of symbols sent by the source and relay, and ${\cal Y}_{r}$ and ${\cal Y}_{d}$ denote the set of symbols received by the relay and destination.  The channel statistics are defined by the probability transition matrix $p(y_{r},y_{d}|x_{s},x_{r})$, where $y_{r}\in{\cal Y}_{r}$, $y_{d}\in{\cal Y}_{d}$, $x_{s}\in{\cal X}_{s}$, and $x_{r}\in{\cal X}_{r}$.

Associated with the channel are sequences of message sets, source and relay encoding functions, and destination decoding functions all indexed by $n$.  Let $M^{(n)}_{s}:=\{1,\ldots, |M^{(n)}_{s}|\}$ denote the source message set.  Let $f^{(n)}_{s}: M^{(n)}_{s}\rightarrow{\bf x}^{(n)}_{s}$ denote the source encoding function, where ${\bf x}^{(n)}_{s}:=\{x^{(i)}_{s}\in{\cal X}_{s}:1\leq i\leq n\}$.  Let $f^{(n)}_{r}:{\bf y}^{(n-1)}_{r}\rightarrow{\bf x}^{(n)}_{r}$ denote the relay encoding function, where ${\bf y}^{(n-1)}_{r}:=\{y^{(i)}_{r}\in{\cal Y}_{r}:1\leq i\leq n\}$ and ${\bf x}^{(n)}_{r}:=\{x^{(i)}_{r}\in{\cal X}_{r}:1\leq i\leq n\}$.  Let $g^{(n)}_{d}:{\bf y}^{(n)}_{d}\rightarrow\hat{M}^{(n)}_{s}$ denote the destination decoding function, where ${\bf y}^{(n)}_{d}:=\{y^{(i)}_{d}\in{\cal Y}_{d}:1\leq i\leq n\}$ and $\hat{M}^{(n)}_{d}:=\{1,\ldots, |M^{(n)}_{s}|\}$.  The estimated message set at the destination $\hat{M}^{(n)}_{d}$ has the same size as the source message set $M^{(n)}_{s}$.

The destination decoding function $g^{(n)}_{d}$ creates a decoding set $A^{(n)}_{m}\in{\cal Y}^{(n)}_{d}$ associated with every message $m\in\hat{M}^{(n)}_{d}$.  We will assume the decoding sets are disjoint.  If message $m\in M^{(n)}_{s}$ is selected by the source, the source encoding function $f^{(n)}_{r}$ sends the sequence ${\bf x}^{(n)}_{s}(m)$.  This sequence passes through the channel and generates the random received sequence ${\bf Y}^{(n)}_{r}$ at the relay, which (via $f^{(n)}_{r})$ generates the random relay transmitted sequence ${\bf X}^{(n)}_{r}$.  The source and relay sequences ${\bf x}^{(n)}_{s}$ and ${\bf X}^{(n)}_{r}$ interact in the channel to produce the random sequence ${\bf Y}^{(n)}_{d}$ received at the destination.  The probability of the destination correctly decoding the message is:
\begin{align}
\label{decodingsetprob}
\mu_{Y_{d}|{\bf x}^{(n)}_{s}(m)}(A^{(n)}_{m}).
\end{align}

The measure $\mu_{Y_{d}|{\bf x}^{(n)}(s)}(\cdot)$ in (\ref{decodingsetprob}) is induced by the probability transition matrix $p(y_{r},y_{d}|x_{s},x_{r})$ and the relay encoding function $f^{(n)}_{r}$ producing ${\bf X}^{(n)}_{r}$.  The capacity of the relay channel is the supremum of all rates $R\doteq|M^{(n)}_{s}|$, with a corresponding sequence of encoding/decoding functions $\{f^{(n)}_{s}, f^{(n)}_{r}, g^{(n)}_{d}:n\in\mathbb{N}\}$ that satisfies:
\begin{align}
    \label{vanishingdecodingerror}
    \lim_{n\rightarrow\infty}\min_{m\in M^{(n)}_{s}}\mu_{Y_{d}|{\bf x}^{(n)}_{s}(m)}(A^{(n)}_{m})=1.
\end{align}

We will introduce some preliminary definitions.  Let $\hat{\cal Y}_{r}$ denote an auxiliary symbol set at the relay, used for building relay compressions.  We can restrict the cardinality of $\hat{\cal Y}_{r}$ (so the search space isn't infinite):
\begin{align}
    |\hat{\cal Y}_{r}|=|{\cal X}_{r}|.
\end{align}

A compress-forward scheme is characterized by a set consisting of an input distribution $p_{\text{cf}}(x_{s},x_{r})$ over ${\cal X}_{s}\times{\cal X}_{r}$, a compression prior $p_{\text{cf}}(\hat{y}_{r}|x_{r})$ over $\hat{{\cal Y}}_{r}\times{\cal X}_{r}$, and a reverse-compression channel $p_{\text{cf}}(y_{r}|\hat{y}_{r},x_{r})$ over ${\cal Y}_{r}\times\hat{\cal Y}_{r}\times{\cal X}_{r}$.  
\begin{align}
    \label{compressforwardinputset}
    {\cal P}_{\text{cf}}:=\{p_{\text{cf}}(x_{s},x_{r}),p_{\text{cf}}(y_{r}|\hat{y}_{r},x_{r}),p_{\text{cf}}(\hat{y}_{r}|x_{r})\}.
\end{align}
The elements in ${\cal P}_{\text{cf}}$ can interact through the channel $p(y_{r},y_{d}|x_{s},x_{r})$ which partially explains why we don't have three degrees of freedom (the full explanation requires a codebook-level analysis of the scheme).  It is useful to define:
\begin{align}
    \label{constraintoncorrelatedset}
    p_{\text{for}}(y_{r}|\hat{y}_{r},x_{r})&:=\frac{\sum_{x_{s},y_{d}}p_{\text{cf}}(x_{s},x_{r})\times p_{\text{cf}}(\hat{y}_{r}|x_{r})\times p(y_{r},y_{d}|x_{s},x_{r})}{\sum_{x_{s},\hat{y}_{r},y_{d}}p_{\text{cf}}(x_{s},x_{r})\times p_{\text{cf}}(\hat{y}_{r}|x_{r})\times p(y_{r},y_{d}|x_{s},x_{r})}.
\end{align}
The reverse-compression-channel in (\ref{constraintoncorrelatedset}) is defined by the forward relay channel $p(y_{r},y_{d}|x_{s},x_{r})$ and the input $p_{\text{cf}}(x_{s},x_{r})$. The ``for'' subscript refers to this dependence on the forward relay channel.  We will also define a forward-compression-channel (which is not the same as the forward relay channel) in terms of the reverse-compression-channel.  Since there are two ways of defining the latter (i.e., either via $p_{\text{cf}}(y_{r}|\hat{y}_{r},x_{r})$ or $p_{\text{for}}(y_{r}|\hat{y}_{r},x_{r})$ in (\ref{constraintoncorrelatedset})), there are two ways of defining the former:   
\begin{align}
    \label{forwardcompressionchannelv1}
    p_{\text{for}}(\hat{y}_{r}|y_{r},x_{r})&:=\frac{\sum_{x_{s},x_{r}}p_{\text{for}}(y_{r}|\hat{y}_{r},x_{r})\times p_{\text{cf}}(\hat{y}_{r}|x_{r})\times p_{\text{cf}}(x_{s},x_{r})}{\sum_{x_{s},x_{r},y_{r}}p_{\text{for}}(y_{r}|\hat{y}_{r},x_{r})\times p_{\text{cf}}(\hat{y}_{r}|x_{r})\times p_{\text{cf}}(x_{s},x_{r})},\\
    \label{forwardcompressionchannelv2}
    p_{\text{cf}}(\hat{y}_{r}|y_{r},x_{r})&:=\frac{\sum_{x_{s},x_{r}}p_{\text{cf}}(y_{r}|\hat{y}_{r},x_{r})\times p_{\text{cf}}(\hat{y}_{r}|x_{r})\times p_{\text{cf}}(x_{s},x_{r})}{\sum_{x_{s},x_{r},y_{r}}p_{\text{cf}}(y_{r}|\hat{y}_{r},x_{r})\times p_{\text{cf}}(\hat{y}_{r}|x_{r})\times p_{\text{cf}}(x_{s},x_{r})}.
\end{align}
Finally, there are two ways of defining the joint distribution, either through (\ref{forwardcompressionchannelv1}) or (\ref{forwardcompressionchannelv2}):
\begin{align}
    \label{forwardforwardchannel}
    p_{\text{for}}(x_{s},x_{r},\hat{y}_{r},y_{r},y_{d})&:=p_{\text{for}}(\hat{y}_{r}|y_{r},x_{r})\times p(y_{r},y_{d}|x_{s},x_{r})\times p_{\text{cf}}(x_{s},x_{r}),\\
    \label{reverseforwardchannel}
    p_{\text{cf}}(x_{s},x_{r},\hat{y}_{r},y_{r},y_{d})&:=p_{\text{cf}}(\hat{y}_{r}|y_{r},x_{r})\times p(y_{r},y_{d}|x_{s},x_{r})\times p_{\text{cf}}(x_{s},x_{r}).
\end{align}
The interdependence among the elements of ${\cal P}_{\text{cf}}$ creates a dichotomy in the space.  We have two sets of compress-forward distributions: the independent set ${\cal P}_{\text{in}}\subset{\cal P}_{\text{cf}}$ and the correlated set ${\cal P}_{\text{cor}}\subset{\cal P}_{\text{cf}}$.  The independent condition is:
\begin{align}
    \label{independentconditionv1}
    p_{\text{cf}}(x_{s},x_{r})&:=p_{\text{cf}}(x_{s})p_{\text{cf}}(x_{r}).
\end{align}
The correlated condition is:
\begin{align}
    \label{correlatedconditionv1}
    p_{\text{cf}}(y_{r}|\hat{y}_{r},x_{r})&:=p_{\text{for}}(y_{r}|\hat{y}_{r},x_{r}),
\end{align}
where $p_{\text{for}}(y_{r}|\hat{y}_{r},x_{r})$ is defined in (\ref{constraintoncorrelatedset}).  The constraint in (\ref{correlatedconditionv1}) is the zero-one circuit-consistency law in a nutshell.  It says that a correlated input $p_{\text{cf}}(x_{s},x_{r})$ requires alignment between the forward (relay) channel and the reverse (compression) channel.  This alignment takes the form of equality between the two different ways of defining the reverse-compression-channel: one as a degree of freedom, the other as a dependent on the input, the prior, and the relay channel itself.  If the forward and reverse channels are not aligned then the input distribution must be independent (this will need to be proved).  We have:
\begin{align}
    \label{independentconditionv2}
    {\cal P}_{\text{ind}}:=\{(p_{\text{cf}}(x_{s},x_{r}),p_{\text{cf}}(y_{r}|\hat{y}_{r},x_{r}),p_{\text{cf}}(x_{s},x_{r})):(\ref{independentconditionv1})\text{ is satisfied}\},\\
    \label{correlatedconditionv2}
    {\cal P}_{\text{cor}}:=\{(p_{\text{cf}}(x_{s},x_{r}),p_{\text{cf}}(y_{r}|\hat{y}_{r},x_{r}),p_{\text{cf}}(x_{s},x_{r})):(\ref{correlatedconditionv1})\text{ is satisfied}\}.
\end{align}
We will also need to further refine ${\cal P}_{\text{ind}}$ and ${\cal P}_{\text{cor}}$ by adding feasibility constraints expressed with mutual information and entropy terms.  However, these terms are going to depend on the joint distributions in (\ref{forwardforwardchannel}) and (\ref{reverseforwardchannel}).  Our first feasibility condition will ensure the compressions cover the space of relay observations:
\begin{align}
    \label{feasibilityset1}
    H_{\text{cf}}(Y_{r}|X_{r})\geq H_{\text{for}}(Y_{r}|X_{r}),
\end{align}
where the entropy terms on the left and right of (\ref{feasibilityset1}) are evaluated via the distributions in (\ref{reverseforwardchannel}) and (\ref{forwardforwardchannel}) respectively.  Our second set of feasibility conditions ensures the destination can recover the compression:
\begin{align}
    \label{feasibilityset2forindependent}
    I_{\text{cf}}(\hat{Y}_{r};Y_{r}|X_{r})&<I_{\text{cf}}(X_{r};Y_{d}),\\
    \label{feasibilityset2forcorrelated}
    I_{\text{for}}(\hat{Y}_{r};Y_{r}|X_{r})&<I_{\text{for}}(\hat{Y}_{r},X_{r};Y_{d}),
\end{align}
where the mutual information terms in (\ref{feasibilityset2forindependent}) and (\ref{feasibilityset2forcorrelated}) are evaluated via (\ref{reverseforwardchannel}) and (\ref{forwardforwardchannel}) respectively.
We have the following:
\begin{align}
    \label{feasiblecorrelateddistribution}
    {\cal P}_{\text{cor-feasible}}&:=\{(p(x_{s},x_{r}),p(y_{r}|\hat{y}_{r},x_{r}),p(\hat{y}_{r}|x_{r}))\in{\cal P}_{\text{ind}}:(\ref{feasibilityset1})\text{ and }(\ref{feasibilityset2forcorrelated})\text{ are satisfied}\},\\
    \label{feasibleindependentdistribution}
    {\cal P}_{\text{ind-feasible}}&:=\{(p(x_{s},x_{r}),p(y_{r}|\hat{y}_{r},x_{r}),p(\hat{y}_{r}|x_{r}))\in{\cal P}_{\text{cor}}:(\ref{feasibilityset1})\text{ and }(\ref{feasibilityset2forindependent})\text{ are satisfied}\}.
\end{align}
Note that (\ref{feasibilityset1}) is trivially satisfied by (\ref{feasiblecorrelateddistribution}) because of (\ref{correlatedconditionv1}).  Now we define three rates: an uncoordinated-compress-forward rate, a coordinated-compress-forward rate, and a decode-forward rate.  First the uncoordinated-compress-forward rate:
\begin{align}
    \label{uncoordinatedcompressforward}
    R_{\text{U-CF}}&:=\max_{(p(x_{s},x_{r}),p(\hat{y}_{r}|x_{r}),p(y_{r}|\hat{y}_{r},x_{r}))\in{\cal P}_{\text{ind-feasible}}}I_{\text{cf}}(X_{s};\hat{Y}_{r},Y_{d}|X_{r}).
\end{align}
Second, the coordinated-compress-forward rate:
\begin{align}
    \label{coordinatedcompressforward}
    R_{\text{C-CF}}&:=\max_{(p(x_{s},x_{r}),p(\hat{y}_{r}|x_{r}),p(y_{r}|\hat{y}_{r},x_{r}))\in{\cal P}_{\text{cor-feasible}}}I_{\text{for}}(X_{s};\hat{Y}_{r},Y_{d}|X_{r}).
\end{align}
The mutual-information terms in (\ref{uncoordinatedcompressforward}) and (\ref{coordinatedcompressforward}) are evaluated via the distribution in (\ref{forwardforwardchannel}) and (\ref{reverseforwardchannel}) respectively.  Finally, define the decode-forward rate:
\begin{align}
    \label{RDF}
    R_{\text{DF}}&:=\max_{p(x_{s},x_{r})}\min\{I(X_{s};Y_{r}|X_{r}),I(X_{s},X_{r};Y_{d})\},
\end{align}
where the joint distribution used in (\ref{RDF}) is  $p(x_{s},x_{r},y_{r},y_{d})=p(y_{r},y_{d}|x_{s},x_{r})p(x_{s},x_{r})$.  
We have the following theorem:
\begin{theorem}
    \label{theoremone}
    The capacity of the relay channel is $\max\{R_{\text{DF}}, R_{\text{U-CF}},R_{\text{C-CF}}\}$.
\end{theorem}

\section{Proof of Theorem 4}
\label{sec:proofoftheorem1}
If the channel capacity $C<\max_{p(x_{s},x_{r})}I(X_{s};Y_{r}|X_{r})$ then the classical decode-forward scheme achieves it. It remains to show that $C=\max\{R_{\text{C-CF}},R_{\text{U-CF}}\}$ when $C>\max_{p(x_{s},x_{r})}I(X_{s};Y_{r}|X_{r})$ (i.e., when the relay cannot decode the source message).  

Here is the architecture of the proof.  We restrict our attention to block encoding/decoding schemes and ``good'' source codebooks that fill up the space of received sequences at the relay.  Then we develop the following genie-aided argument to separate the relay $\rightarrow$ destination channel from the source $\rightarrow$ (relay, destination) channel.  If a genie reveals the relay transmission to the destination, then the destination decoder can focus on source codewords that significantly contribute to the probability mass of the relay transmission.  We show that the rate on the genie-scoped version of the source codebook is (approximately) $R-I(X_{s};\hat{Y}_{r}|X_{r})$.  We use the sphere-packing argument to get this rate.  From here, we get the compression-cut set bound $R<I(X_{s};\hat{Y}_{r},Y_{d}|X_{r})$.  

Finally, we justify our initial decision to examine good block encoding/decoding schemes.  In particular, there are two regimes of interest: a block-coding regime and a symbol-by-symbol regime.  We also defend the coordinated/uncoordinated dichotomy in the compress-forward setting and introduce feasibility constraints to compensate the genie.  On the achievability side, we introduce the coordinated-compress-forward scheme (which replaces the partial decode-forward scheme in \cite{Cover1979}).

\subsection{Outer-Bound: Setup}
{\bf SU-1)}  We focus on source and relay encoding functions that use block encoding.  In the block encoding scheme, transmission occurs over $B$ blocks of $n$ channel uses.  In each block, there is a source encoding function: 
\begin{align}
\label{def:sourceblockencoding}
f^{(n,B)}_{s}:M^{(n,B)}_{r}\times M^{(n,B)}_{s}\rightarrow{\bf x}^{(n,B)}_{s},    
\end{align}
where $M^{(n,B)}_{s}$ denotes the message set assigned to block $B$ and $M^{(n,B)}_{r}$ is a relay message set that indexes all possible sequences the relay can transmit in block $B$.  Associated with the source encoding function in (\ref{def:sourceblockencoding}) is a relay encoding function in block $B+1$:  

\begin{align}
\label{def:relayblockencoding}
f^{(n,B+1)}_{r}:{\bf y}^{(n,B)}_{r}\rightarrow{\bf x}^{(n,B+1)}_{r},    
\end{align}

We assume the source message set satisfies $|M^{(n,B)}_{s}|\doteq R$.  More generally, the source rate can vary with $n$ and $B$, we fold this scenario into our analysis of the block-coding regime and symbol-by-symbol regime in the wrap-up.  We do allow the rate of the relay message set $|M^{(n,B)}_{r}|\doteq R^{(n,B)}_{r}$ to vary with $n$ and $B$.  Apart from the block encoding structure in (\ref{def:sourceblockencoding}) and (\ref{def:relayblockencoding}) and the assumption on the size of the source message set, we assume no other structure on the source codebook and relay codebooks (i.e., they are not i.i.d).

{\bf SU-2)} In general, every relay codeword transmitted in block $B$ is pinned to a (possibly unique) source codebook in block $B$.  Let us freeze $n$ and $B$ and examine an experimental instance in which the relay sends some codeword ${\bf x}^{(n,B)}_{r}(m_{r})$.  Associated with this codeword is a source codebook:
\begin{align}
    \label{sequenceofcodebooks}
    {\cal C}^{(n,B)}_{s}(R)|_{m_{r}}:=\{{\bf x}^{(n,B)}_{s}(m_{s},m_{r}):m_{s}\in M^{(n,B)}_{s}\}.
\end{align}
Each codeword pair $({\bf x}^{(n,B)}_{s}(m_{s},m_{r}),{\bf x}^{(n,B)}_{r}(m_{r}))$ has a joint type ${\bf p}^{(n,B)}_{s,r}$ in the simplex over ${\cal X}_{s}\times{\cal X}_{r}$.  We can restrict our attention to the codeword pairs (pinned to ${\bf x}_{r}(m_{r})$) with ``significant'' types which have a non-vanishing fraction of the mass.  In our experimental instance, the source is going to pick a source codeword such that $({\bf x}^{(n,B)}_{s}(m_{s},m_{r}),{\bf x}^{(n,B)}_{r}(m_{r}))$ has joint type ${\bf p}^{(n,B)}_{s,r}$.  With high probability, this type will be significant.  We can then assume that the codewords in ${\cal C}^{(n)}_{s}(R)|_{m_{r}}$ with ${\bf x}^{(n,B)}_{r}$ share the same significant type ${\bf p}^{(n,B)}_{s,r}$ (because this subcodebook needs to be decodeable if the larger subcodebook is decodeable).  We are repeating the same argument from section \ref{sec:pointtopointchannel} except now for codeword pairs.  Note that the index pair $(m_{s},m_{r})$ in our experiment is sensitive to both $n$ and $B$.  Now we unfreeze $n$ (but keep $B$ fixed). We take a subsequence $n_{k}$  so that: 
\begin{align}
    \label{limitoftypes}
    {\bf p}_{s,r}=\lim_{k\rightarrow\infty}{\bf p}^{(n_{k},B)}_{s,r}.
\end{align}
The subsequence (\ref{limitoftypes}) exists because the simplex is compact.

{\bf SU-3)}  Suppose a genie reveals the relay codeword ${\bf x}^{(n_{k},B+1)}_{r}$ to the destination at the end of block $B$.  Strictly speaking, every codeword in the source codebook ${\cal C}^{(n_{k},B)}_{s}(R)|_{m_{r}}$ contributes to the probability mass of ${\bf x}^{(n_{k},B+1)}_{r}$.  However, some codewords contribute more to this mass than others.  We ignore a subset of codewords whose collective contribution is a vanishing fraction of this mass to get a genie-scoped codebook $\hat{\cal C}^{(n_{k},B)}_{s}(R^{(B)}_{g})\subseteq{\cal C}^{(n_{k},B)}_{s}(R)$.  Implicit in this formulation is the assertion that the codebook rate is converging to $R^{(B)}_{g}$.  Let $p(x_{s},_{r})$ denote the  distribution over ${\cal X}_{s}\times{\cal X}_{r}$ that satisfies:
\begin{align}
    \label{typesconvergetodistribution}
    p(x_{s},x_{r}):={\bf p}_{s,r}.
\end{align}
That is, $p(x_{s},x_{r})$ corresponds to the point on the simplex to which our subsequence of types is converging.  For the destination to pick the actual transmitted source codeword out of $\hat{\cal C}^{(n_{k},B)}_{s}(R^{(B)}_{g})$, we have the following constraint:  
\begin{align}
    \label{classicbound}
    R^{(B)}_{g}<I(X_{s};Y_{d}|X_{r}),
\end{align}
where the mutual information in (\ref{classicbound}) is defined by the distribution $p(y_{r},y_{d}|x_{s},x_{r})p(x_{s},x_{r})$ and $p(x_{s},x_{r})$ is defined by (\ref{limitoftypes}) and (\ref{typesconvergetodistribution}).  In the classical compress-forward scheme using i.i.d source codebooks and i.i.d compressions, one can show that each relay codeword maps one-to-one with a compression $\hat{\bf y}_{r}$: 
\begin{align}
\label{genieconstraint}
R^{(B)}_{g}<R-I(X_{s};\hat{Y}_{r}|X_{r}).   
\end{align}
Here, the codewords are structured and the relay encoding function is arbitrary.  Our main task in the next section is to show that (\ref{genieconstraint}) is active even in the structured, arbitrary setting.  

\subsection{Outer-Bound: Main Argument}
{\bf MA-1)} Over the next few steps we develop a notion of good source codebooks to get an outer-bound.  We have already decided that the relay sends ${\bf x}^{(n_{k},B)}_{r}$ in block $B$ which is indexed by $m_{r}$. Now suppose the source sends the codeword ${\bf x}^{(n_{k},B)}_{s}$ in block $B$. Let us fix ${\bf y}^{(n_{k},B)}_{r}$ such that: 
\begin{align}
\label{orbitofacodeword}
{\bf y}^{(n_{k},B)}_{r}\in T^{(n_{k},B)}_{*}(Y_{r}|{\bf x}^{(n_{k},B)}_{r},{\bf x}^{(n_{k},B)}_{s}).    
\end{align}
We will say that the ${\bf y}^{(n_{k},B)}_{r}$ in (\ref{orbitofacodeword}) is in the orbit of ${\bf x}^{(n_{k},B)}_{s}$.
Conditioned on $({\bf x}^{(n_{k},B)}_{s},{\bf x}^{(n_{k},B)}_{r})$ this relay observation has approximate mass:
\begin{align}
    \label{approximatemassofaobservationintheorbit}
    2^{-n_{k}H(Y_{r}|X_{s},X_{r})}.
\end{align}
By ``approximate'', we mean that the quantity may differ in the normalized exponent by some $\gamma^{(n_{k})}$ and the sequence $\{\gamma^{(n_{k})}:k\in\mathbb{N}\}$ is vanishing.  In the case of (\ref{approximatemassofaobservationintheorbit}), this implies the normalized exponent lies in the interval:  
\begin{align}
    [H(Y_{r}|X_{s},X_{r})-\gamma^{(n_{k})},H(Y_{r}|X_{s},X_{r})+\gamma^{(n_{k})}].
\end{align}
The distribution used to get the mass in (\ref{approximatemassofaobservationintheorbit}) is denoted by:
\begin{align}
\label{relaydistributionpercodeword}
\mu_{Y_{r}|{\bf x}^{(n_{k},B)}_{s},{\bf x}^{(n_{k},B)}_{r}}(\hspace{0.5mm}\cdot\hspace{0.5mm}),  
\end{align}
where (\ref{relaydistributionpercodeword}) is generated via $p(y_{r}|x_{s},x_{r})$.  Now after taking subsequences in {\bf SU-2}, we also have a subsequence of codebooks in (\ref{sequenceofcodebooks}) which gives:
\begin{align}
    \label{subsequenceofcodebooks}
    {\cal C}^{(n_{k},B)}_{s}(R)|_{m_{r}}:=\{{\bf x}^{(n_{k},B)}_{s}(m_{s},m_{r}):m_{s}\in M^{(n_{k},B)}_{s}\}.
\end{align}
Every codeword in the codebook in (\ref{subsequenceofcodebooks}) has mass:
\begin{align}
\label{massofacodeword}
\frac{1}{|M^{(n_{k},B)}_{s}|}.    
\end{align}

{\bf MA-2)} Now let us fix some ${\bf y}^{(n_{k},B)}_{r}$ in the orbit of the codebook in (\ref{subsequenceofcodebooks}) so that:
\begin{align}
\label{sourcerelaycodewordorbit}
{\bf y}^{(n_{k},B)}_{r}\in\bigcup_{m\in M^{(n_{k},B)}_{s}}T^{(n_{k},B)}_{*}(Y_{r}|{\bf x}^{(n_{k},B)}_{s}(m),{\bf x}^{(n_{k},B)}_{r}).    
\end{align}
We want to find the mass of ${\bf y}^{(n_{k},B)}_{r}$ after weighting the codewords according to (\ref{massofacodeword}) but we have a problem.  All codewords in (\ref{subsequenceofcodebooks}) contribute to the mass of ${\bf y}^{(n_{k},B)}_{r}$ in one of two different ways, and this division of labor depends on the structure of the codebook which is arbitrary.  We look at the mass of an observation ${\bf y}^{(n_{k},B)}_{r}$ given by the measure (\ref{relaydistributionpercodeword}).  The first way a codeword ${\bf x}^{(n_{k},B)}_{s}$ could contribute to the mass of ${\bf y}^{(n_{k},B)}_{r}$ is through the true channel $p(y_{r}|x_{s},x_{r})$ in which case ${\bf y}^{(n_{k},B)}_{r}$ satisfies (\ref{orbitofacodeword}) and has mass (\ref{approximatemassofaobservationintheorbit}).  The second way is through one of a polynomial number of false channels $p^{\prime}(y_{r}|x_{s},x_{r})$ in which case the conditional mass is given by:
\begin{align}
    \label{massviaafalsechannel}
    2^{-n_{k}(H(Y_{r}|X_{s},X_{r})+D)},
\end{align}
where the $D$ in (\ref{massviaafalsechannel}) is the divergence between $p(y_{r}|x_{s},x_{r})$ and $p^{\prime}(y_{r}|x_{s},x_{r})$.  We define the mass of ${\bf y}^{(n_{k},B)}_{r}$ as the weighted average of the mass coming from all codewords:
\begin{align}
\label{relaydistributionovercodebook}
\mu_{Y_{r}|{\cal C}^{(n_{k},B)}_{s}(R)|_{m_{r}}}({\bf y}^{(n_{k},B)}_{r}):=\sum_{m\in M^{(n_{k},B)}_{s}}\mu_{Y_{r}|{\bf x}^{(n_{k},B)}_{s}(m),{\bf x}^{(n_{k},B)}_{r}}({\bf y}^{(n_{k},B)}_{r})\times\frac{1}{|M^{(n_{k},B)}_{s}|}.    
\end{align}

{\bf MA-3)}  We would like to classify the observation sequences generated by the codebook according to their mass, per the measure in (\ref{relaydistributionovercodebook}).  For some vanishing $\{\epsilon^{(n_{k})}_{I}:k\in\mathbb{N}\}$ define the set:
\begin{align}
    \label{definingthesetofrelayobservations}
    I^{(n_{k},B)}:=\{{\bf y}^{(n_{k},B)}_{r}:2^{-n_{k}(H(Y_{r}|X_{r})+\epsilon^{(n_{k})}_{I})}<\mu_{Y_{r}|{\cal C}^{(n_{k},B)}_{s}(R)|_{m_{r}}}({\bf y}^{(n_{k},B)}_{r})<2^{-n_{k}(H(Y_{r}|X_{r})-\epsilon^{(n_{k})}_{I})}\},
\end{align}
where the entropy $H(Y_{r}|X_{r})$ in (\ref{definingthesetofrelayobservations}) is defined with respect to the distribution $p(y_{r},y_{d}|x_{s},x_{r})p(x_{s},x_{r})$ and $p(x_{s},x_{r})$ is defined in (\ref{typesconvergetodistribution}).  The focus on sequences of codebooks $\{{\cal C}^{(n_{k},B)}_{s}(R)|_{m_{r}}:k\in\mathbb{N}\}$ for which there exists some vanishing $\{\epsilon^{(n_{k})}_{I}:k\in\mathbb{N}\}$ and some fixed $\epsilon>0$ such that:
\begin{align}
    \label{measureofinterval}
    \mu_{Y_{r}|{\cal C}^{(n_{k},B)}_{s}(R)|_{m_{r}}}(I^{(n_{k},B)})>\epsilon.
\end{align}
A sequence of codebooks that satisfies (\ref{measureofinterval}) has a sequence of good subcodebooks, and each good subcodebook has enough codewords from the original to have mass (possibly small mass, but non-vanishing).  Since they have non-vanishing mass, their rates asymptotically converge to $R$.  The observations generated by this  sequence of good source codebooks fill up the relay observation space.  The remainder of our analysis will focus on the sequence of good source subcodebooks, which for convenience, we keep denoting by $\{{\cal C}^{(n_{k},B)}_{s}(R)|_{m_{r}}:k\in\mathbb{N}\}$ instead of further complicating the notation.

{\bf MA-4)}  For any ${\bf y}^{(n_{k},B)}_{r}\in I^{(n_{k},B)}$, we want to find the number of codewords contributing to the mass of ${\bf y}^{(n_{k},B)}_{r}$.  More precisely, we want the normalized exponent on the number of codewords contributing a non-vanishing fraction of the mass of ${\bf y}^{(n_{k},B)}_{r}$ (because technically all codewords are contributing to its mass).  The number of relay observations is approximately:
\begin{align}
    \label{numberofrelayobservations}
    2^{n_{k}H(Y_{r}|X_{r})},
\end{align}
where (\ref{numberofrelayobservations}) follows from (\ref{measureofinterval}).  The number of observations in the orbit of each codeword is approximately:
\begin{align}
    \label{numberofobservationspercodeword}
    2^{n_{k}H(Y_{r}|X_{s},X_{r})}.
\end{align}
Each observation in the orbit of a codeword has approximate conditional mass:
\begin{align}
    \label{massofanobservationintheorbit}
    2^{-n_{k}H(Y_{r}|X_{s},X_{r})}.
\end{align}
Each codeword has the same mass given by (\ref{massofacodeword}) and $\frac{1}{n_{k}}\log_{2}|M^{(n_{k},B)}_{s}|\doteq R$.  Moreover, each ${\bf y}^{(n_{k},B)}_{r}\in I^{(n_{k},B)}$ is getting mass either through the true channel $p(y_{r}|x_{s},x_{r})$ (which gives (\ref{massofanobservationintheorbit})) or at most a polynomial number of false channels each of which give (\ref{massviaafalsechannel}).  It follows from (\ref{numberofrelayobservations})-(\ref{massofanobservationintheorbit}) that the number of codewords contributing a non-vanishing fraction of the mass to each ${\bf y}^{(n_{k},B)}_{r}\in I^{(n_{k},B)}$ is approximately:
\begin{align}
    \label{numberofcodewordscontributingmasstoanobservation}
    2^{n_{k}(R-I(X_{s};Y_{r}|X_{r}))},
\end{align}
where all the entropy terms in (\ref{numberofrelayobservations})-(\ref{numberofcodewordscontributingmasstoanobservation}) are made with respect to the distribution
$p(y_{r},y_{d}|x_{s},x_{r})p(x_{s},x_{r})$ and $p(x_{s},x_{r})$ is given by (\ref{typesconvergetodistribution}).  We are going to play this argument forward into the compression-space.

{\bf MA-5)}  Now each observation ${\bf y}^{(n_{k},B)}_{r}\in I^{(n_{k},B)}$ maps to a relay sequence ${\bf x}^{(n_{k},B+1)}_{r}$ according to the mapping: 
\begin{align}
    \label{relayencodingfunctionsubsequence}
    f^{(n_{k},B+1)}_{r}:{\bf y}^{(n_{k},B)}_{r}\rightarrow{\bf x}^{(n_{k},B+1)}_{r}.
\end{align}
This mapping creates the following relay codebook: 
\begin{align}
    \label{relaycodebook}
    {\cal C}^{(n_{k},B+1)}_{r}(R^{(n_{k},B+1)}_{r})&:=\{{\bf x}^{(n_{k},B+1)}_{r}(m):m\in M^{(n_{k},B+1)}_{r}\}. 
\end{align}
In (\ref{relaycodebook}), we've indexed all the codewords generated by the mapping (\ref{relayencodingfunctionsubsequence}).  The index set $M^{(n_{k},B)}_{r}$ has variable rate $R^{(n_{k},B)}_{r}$ with respect to $n_{k}$ and $B$ because (\ref{relayencodingfunctionsubsequence}) is arbitrary.  We take another subsequence $\{n_{k}:k\in\mathbb{N}\}$ to get a sequence of relay codebooks that satisfies the following two properties:

{\bf P1)} $\lim_{k\rightarrow\infty}R^{(n_{k},B+1)}_{r}=R_{r}$,

{\bf P2)} $R_{r}$ is the smallest rate for which ({\bf P1}) exists.  

Note that since $R^{(n_{k},B+1)}_{r}\in[0,|{\cal X}_{r}|]$, both conditions are well-defined.  Let us define $R_{0}$ as follows:
\begin{align}
    \label{R0}
    R_{0}:=H(Y_{r}|X_{r})-R_{r}
\end{align}
Each codeword ${\bf x}^{(n_{k},B+1)}_{r}\in{\cal C}^{(n_{k},B+1)}_{r}(R^{(n_{k},B+1)}_{r})$ has an inverse image via (\ref{relayencodingfunctionsubsequence}) in the ${\cal Y}^{(n_{k})}_{r}$ space.  Define the inverse image function $f^{-1}_{r}(\cdot)$ as follows:
\begin{align}
    \label{inverseimagefunction}
    f^{-1}_{r}({\bf x}^{(n_{k},B+1)}_{r}):=\{{\bf y}^{(n_{k},B)}_{r}:f^{(n_{k},B+1)}_{r}({\bf y}^{(n_{k},B)}_{r})={\bf x}^{(n_{k},B+1)}_{r}\}.
\end{align}
We are going to filter the relay codebooks by only keeping codewords ${\bf x}^{(n_{k},B+1)}_{r}$ that satisfy the following rule:

{\bf P3)}  For some vanishing $\{\epsilon^{(n_{k})}_{r}:k\in\mathbb{N}\}$:
\begin{align}
    \label{boundsonthesizeoftheinverseimageofacompression}
    R_{0}-\epsilon^{(n_{k})}_{r}<\frac{1}{n_{k}}\log_{2}|f^{-1}_{r}({\bf x}^{(n_{k},B+1)}_{r})|<R_{0}+\epsilon^{(n_{k})}_{r}.
\end{align}
We will keep the same notation ${\cal C}^{(n_{k},B+1)}_{r}(R^{(n_{k},B+1)}_{r})$ to denote the filtered relay codebook for convenience.  Now the measure (\ref{relaydistributionovercodebook}) puts mass on the relay codebook via the mapping (\ref{relayencodingfunctionsubsequence}).  It follows from (P2) that there is an $\epsilon>0$ such that for all $k\in\mathbb{N}$:
\begin{align}
    \label{filtererdcodebookstillhasmass}
    \mu_{Y_{r}|{\cal C}^{(n_{k},B)}_{s}(R)|_{m_{r}}}({\cal C}^{(n_{k},B+1)}_{r}(R^{(n_{k},B+1)}_{r}))>\epsilon.
\end{align}
The significance of (\ref{filtererdcodebookstillhasmass}) is that the filtered relay codebook still gets mass from the source codebook (i.e., a source subcodebook with rate close to $R$), and therefore cannot be ignored.  

To close out this step, we note that each relay codeword ${\bf x}^{(n_{k},B+1)}_{r}$ is a compression because the mapping (\ref{relayencodingfunctionsubsequence}) is surjective.  We will build a codebook of compressions on top of the relay codeword ${\bf x}^{(n_{k},B)}_{r}$ sent in block $B$.  We select a compression prior distribution $p(\hat{y}_{r}|x_{r})$.  For each codeword ${\bf x}^{(n_{k},B+1)}_{r}\in\hat{\cal C}^{(n_{k},B+1)}_{r}(R_{r})$, we generate a compression codeword $\hat{\bf y}^{(n_{k},B)}_{r}$ i.i.d according to $p(\hat{y}_{r}|x_{r})$.  This gives us a compression codebook: 
\begin{align}
\label{def:compressioncodebook}
\hat{\cal C}^{(n_{k},B)}_{r}(R_{r}):=\{\hat{\bf y}^{(n_{k},B)}_{r}(m):m\in\hat{M}^{(n_{k},B)}_{r}\},    
\end{align}
where every ${\bf y}^{(n_{k},B)}_{r}\in I^{(n_{k},B)}$ maps to a relay compression $\hat{\bf y}^{(n_{k},B)}_{r}$ via some mapping:
\begin{align}
    \label{def:compressionmapping}
    \hat{f}^{(n_{k},B)}_{r}:{\bf y}^{(n_{k},B)}_{r}\rightarrow\hat{\bf y}^{(n_{k},B)}_{r}.
\end{align}

{\bf MA-6)}  Consider the triple $({\bf x}^{(n_{k},B)}_{r},{\bf y}^{(n_{k},B)}_{r},\hat{\bf y}^{(n_{k},B)}_{r})$ induced by (\ref{def:compressionmapping}).  This triple has a conditional type $p^{(n_{k})}(y_{r}|\hat{y}_{r},x_{r})$ in the simplex of reverse-compression-channel distributions over ${\cal X}_{r}\times{\cal Y}_{r}\times\hat{\cal Y}_{r}$.  Associated with this type is the joint distribution:
\begin{align}
    \label{typeinducedbycompressionmapping}
    p^{(n_{k})}(x_{s},x_{r},y_{r},\hat{y}_{r})=p^{(n_{k})}(y_{r}|\hat{y}_{r},x_{r})p(\hat{y}_{r}|x_{r})p(x_{s},x_{r}),
\end{align}
where the compression prior $p(\hat{y}_{r}|x_{r})$ was selected in MA-5 to get the compression codebook (\ref{def:compressioncodebook}) and $p(x_{s},x_{r})$ comes from (\ref{typesconvergetodistribution}).  Associated with (\ref{typeinducedbycompressionmapping}) is a conditional entropy:
\begin{align}
    \label{entropyoftype}
    H^{(n_{k})}(Y_{r}|\hat{Y}_{r},X_{r}).
\end{align}
Now the types in (\ref{typeinducedbycompressionmapping}) are arbitrary because the mapping (\ref{def:compressionmapping}) is also arbitrary, but nevertheless, there are things we can say.  For instance, we know that there can't be too many types where $H^{(n_{k})}(Y_{r}|\hat{Y}_{r},X_{r})<R_{0}-\epsilon^{(n_{k})}_{r}$ because otherwise $I^{(n_{k},B)}$ won't be covered.  In fact, we can say more.  We don't really care about the specific channel types induced by (\ref{def:compressionmapping}).  Rather, we want to know how many codewords contribute to the mass of a compression.  The inverse image of each compression is size-wise approximately the same via (\ref{boundsonthesizeoftheinverseimageofacompression}) and each observation has about the same mass via (\ref{definingthesetofrelayobservations}).  So for the purposes of understanding the mass of a compression, we can assume that the channel types converge to some $p(y_{r}|\hat{y}_{r},x_{r})$ whose entropy satisfies:
\begin{align}
    \label{conditiononchanneltypes}
    H(Y_{r}|\hat{Y}_{r},X_{r})=R_{0}.
\end{align}
Given $p(y_{r}|\hat{y}_{r},x_{r})$ we define the forward-compression-channel $p(\hat{y}_{r}|y_{r},x_{r})$ as in (\ref{forwardcompressionchannelv2}).  The forward-compression-channel in projects mass into the compression space:
\begin{align}
    \label{compressionspace}
    T^{(n_{k},B)}_{*}(\hat{Y}_{r}|{\bf x}^{(n_{k},B)}_{r}).
\end{align}
Moreover, (\ref{measureofinterval}) implies there are approximately $2^{n_{k}H(Y_{r}|X_{r})}$ observations each having approximate mass $2^{-n_{k}H(Y_{r}|X_{r})}$.  Therefore every $\hat{\bf y}^{(n_{k},B)}_{r}\in T^{(n_{k},B)}_{*}(\hat{Y}_{r}|{\bf x}^{(n_{k},B)}_{r})$ has approximate mass:
\begin{align}
    \label{massonacompression}
    2^{-n_{k}H(\hat{Y}_{r}|X_{r})}.
\end{align}
It is important to note that the compression codebook $\hat{\cal C}^{(n_{k},B)}_{r}(R_{r})$ is much smaller in size than the compression space in (\ref{compressionspace}).  Most of the compressions in (\ref{compressionspace}) are ``virtual'' compressions or ``dummy'' compressions.  However:
\begin{align}
    \label{compressioncodebookisinthecompressionspace}
    \hat{\cal C}^{(n_{k},B)}_{r}(R_{r})\subseteq T^{(n_{k},B)}_{*}(\hat{Y}_{r}|{\bf x}^{(n_{k},B)}_{r}),
\end{align}
which means every $\hat{\bf y}^{(n_{k},B)}_{r}\in\hat{\cal C}^{(n_{k},B)}_{r}(R_{r})$ also has approximate mass given by (\ref{massonacompression}).  Now fix ${\bf x}^{(n_{k},B)}_{s}\in{\cal C}^{(n_{k},B)}_{s}(R)$ and consider the compression orbit of ${\bf x}^{(n_{k},B)}_{s}$:
\begin{align}
    \label{compressionorbit}
    \{\hat{\bf y}^{(n_{k},B)}_{r}:\hat{\bf y}^{(n_{k},B)}_{r}\in T^{(n_{k},B)}_{*}(\hat{Y}_{r}|{\bf y}^{(n_{k},B)}_{r},{\bf x}^{(n_{k},B)}_{r}),{\bf y}^{(n_{k},B)}_{r}\in T^{(n_{k},B)}_{*}(Y_{r}|{\bf x}^{(n_{k},B)}_{s},{\bf x}^{(n_{k},B)}_{r})\}.
\end{align}
According to the Markov Lemma, the number of compressions in (\ref{compressionorbit}) is approximately:
\begin{align}
    \label{volumeofcompressionsincompressionorbit}
    2^{n_{k}H(\hat{Y}_{r}|X_{s},X_{r})},
\end{align}
and every $\hat{\bf y}^{(n_{k},B)}_{r}$ in (\ref{compressionorbit}) conditioned on ${\bf x}^{(n_{k},B)}_{s}$ has approximate mass:
\begin{align}
    \label{massofeachcompression}
    2^{-n_{k}H(\hat{Y}_{r}|X_{s},X_{r})}.
\end{align}
From the same sphere-packing argument in (MA-4) applied to (\ref{massonacompression}), (\ref{volumeofcompressionsincompressionorbit}), and (\ref{massofeachcompression}), the number of codewords contributing a non-vanishing fraction of the mass of each $\hat{\bf  y}^{(n_{k},B)}_{r}\in T^{(n_{k},B)}_{*}(\hat{Y}_{r}|{\bf x}^{(n_{k},B)}_{r})$ is approximately:
\begin{align}
    \label{codewordsmappingtoacompression}
    2^{n_{k}(R-I(X_{s};\hat{Y}_{r}|X_{r}))}.
\end{align}
The expression in (\ref{codewordsmappingtoacompression}) applies to all compressions in $T^{(n_{k},B)}_{*}(\hat{Y}_{r}|{\bf x}^{(n_{k},B)}_{r})$, including those in $\hat{\cal C}^{(n_{k},B)}_{r}(R_{r})$.  This gives us the expression in (\ref{genieconstraint}) which leads to the compression-cut-set mutual-information bound in (\ref{uncoordinatedcompressforward}) and (\ref{coordinatedcompressforward}).

\subsection{Outer-Bound: Wrap-Up}
{\bf WU-1)}  We justify the dichotomy in the space of compress-forward schemes which only allows one degree-of-freedom between the input distribution $p(x_{s},x_{r})$ and the reverse-compression-channel  $p(y_{r}|\hat{y}_{r},x_{r})$.  Notice that if:
\begin{align}
    R_{0}<H(Y_{r}|X_{s},X_{r}),
\end{align}
then the conditional probability of a codeword landing on any individual compression vanishes, which is asymptotically equivalent to uniformly selecting any compression in the compression-codebook $\hat{\cal C}^{(n_{k},B)}_{r}(R_{r})$.  This gives an independent input distribution $p(x_{s},x_{r})=p(x_{s})p(x_{r})$.  More generally, if the mapping $p(y_{r}|\hat{y}_{r},x_{r})$ is not aligned with the channel as per (\ref{constraintoncorrelatedset}) then any pair $({\bf y}^{(n_{k},B)}_{r},\hat{\bf y}^{(n_{k},B)}_{r})$ induced by the mapping has vanishing probability (conditioned on the codeword), which again is equivalent to uniformly selecting a compression.

{\bf WU-2)} Next we justify the assumption made to focus on codebooks that satisfy (\ref{measureofinterval}) (i.e., ``good'' codebooks).  For some fixed $\epsilon_{0}>0$, define:
\begin{align}
    \label{definingthecomplementtothesetofrelayobservations}
    \bar{I}^{(n_{k},B)}:=\{{\bf y}^{(n_{k},B)}_{r}:\mu_{Y_{r}|{\cal C}^{(n_{k},B)}_{s}(R)|m_{r}}({\bf y}^{(n_{k},B)}_{r})<2^{-n_{k}(H(Y_{r}|X_{r})-\epsilon_{0})}\}.
\end{align}
If (\ref{measureofinterval}) is violated, then for some $\epsilon_{0},\epsilon>0$ we have:
\begin{align}
    \label{measureofthecomplement}
    \mu_{Y_{r}|{\cal C}^{(n_{k},B)}_{s}(R)|_{m_{r}}}(\bar{I}^{(n_{k},B)})>\epsilon.
\end{align}
Following the same logical chain in (MA-4), we see that (\ref{definingthecomplementtothesetofrelayobservations}) and (\ref{measureofthecomplement}) imply that there is a tinybook where the  number of codewords contributing a non-vanishing fraction of the mass to each ${\bf y}^{(n_{k},B)}_{r}\in\bar{I}^{(n_{k},B)}$ is at least:
\begin{align}
    2^{n_{k}(R+\epsilon_{0}-I(X_{s};Y_{r}|X_{r}))}.
\end{align}
This reduces the compression-cut-set mutual-information bound in (\ref{uncoordinatedcompressforward}) and (\ref{coordinatedcompressforward}) by at least $\epsilon_{0}$.  One wrinkle to consider is that (\ref{measureofthecomplement}) implies we could reduce the compression rate $R_{r}$ while keeping $R_{0}$ fixed to cover $\bar{I}^{(n_{k},B)}$.  This could help us if the relay$\rightarrow$destination channel is weak.  But we could also just simply increase $R_{0}$ by at most $\epsilon_{0}$ instead with no worse effect on the compression-cut-set bound.  So there's no advantage to considering not-good codebooks.

{\bf WU-3)} Now we address the feasibility constraints that compensate for the genie.  Since there is no genie, the destination needs to decode the relay codeword.  Without coordination between the source and relay, the destination can only rely on the relay transmission.  With coordination, the destination can use both the source and relay transmission.  Since $R_{r}=I(\hat{Y}_{r};Y_{r}|X_{r})$, this explains both feasibility constraints in (\ref{uncoordinatedcompressforward}) and (\ref{coordinatedcompressforward}).  

There is one wrinkle. The relay sends the compression to the destination in a different block from the one in which the compression is decided.  This means there is some tension in the choice of input $p(x_{s},x_{r})$ across blocks; we are simultaneously trying to minimize $I(X_{s};\hat{Y}_{r}|X_{r})$ in one block and maximize $I(X_{r};Y_{d})$ in the next.  The question is whether time-sharing different input distributions across blocks could lead to higher-rates.  However, since all the mutual-information terms are concave functions of the input distributions, there is no advantage to time-sharing.

{\bf WU-4)}  Finally, we relax the block-coding assumption.  Even without this assumption, there is still a source codebook ${\cal C}^{(n)}_{s}(R):=\{{\bf x}^{(n)}_{s}(m):m\in M^{(n)}_{s}\}$ induced by the source encoding function $f^{(n)}_{s}$.  For any codeword ${\bf x}^{(n)}_{s}\in{\cal C}^{(n)}_{s}(R)$, we have the measure $\mu_{Y_{r}|{\bf x}^{(n)}_{s}}(\hspace{1mm}\cdot\hspace{1mm})$ induced by $p(y_{r},y_{d}|x_{s},x_{r})$.  We use this measure to define a codebook measure:
\begin{align}
    \label{codebookmeasurenotblockencoding}
    \mu_{Y_{r}|{\cal C}^{(n)}_{s}(R)}({\bf y}^{(n)}_{r})=\sum_{m\in M^{(n)}_{s}}\mu_{Y_{r}|{\bf x}^{(n)}_{s}(m)}({\bf y}^{(n)}_{r})\times\frac{1}{|M^{(n)}_{s}|}.
\end{align}
For any sequence of sets $\{O^{(n)}_{Y_{r}}\subseteq{\cal Y}^{(n)}_{r}:n\in\mathbb{N}\}$, there is a corresponding sequence of rates:
\begin{align}
    \label{observationrate}
    R^{(n)}_{Y_{r}}:=\frac{1}{n}\log_{2}|O^{(n)}_{Y_{r}}|.
\end{align}
We are interested in sequences of sets $\{O^{(n)}_{Y_{r}}\subseteq{\cal Y}^{(n)}_{r}:n\in\mathbb{N}\}$ for which:
\begin{align}
    \label{probabilitygoestoone}
    \lim_{n\rightarrow\infty}\mu_{Y_{r}|{\cal C}^{(n)}_{s}(R)}(O^{(n)}_{Y})=1.
\end{align}
For any sequence of sets $\{O^{(n)}_{Y_{r}}\subseteq{\cal Y}^{(n)}_{r}:n\in\mathbb{N}\}$ satisfying (\ref{probabilitygoestoone}), there is a subsequence $\{n_{k}:k\in\mathbb{N}\}$ such that:
\begin{align}
    \label{subsequence}
    \lim_{k\rightarrow\infty}R^{(n_{k})}_{Y}=R_{Y_{r}}.
\end{align}
For any codeword ${\bf x}^{(n)}_{s}\in{\cal C}^{(n)}_{s}(R)$, we have the distribution $\mu_{X_{r}|{\bf x}^{(n)}_{s}}(\hspace{1mm}\cdot\hspace{1mm})$ induced by $p(y_{r},y_{d}|x_{s},x_{r})$ and the relay encoding function $f^{(n)}_{r}$.  As in (\ref{codebookmeasurenotblockencoding}), we can use this measure to define a codebook measure:
\begin{align}
    \label{codebookmeasureforrelaytransmission}
    \mu_{X_{r}|{\cal C}^{(n)}_{s}(R)}({\bf x}^{(n)}_{r})=\sum_{m\in M^{(n)}_{s}}\mu_{X_{r}|{\bf x}^{(n)}_{s}(m)}({\bf x}^{(n)}_{r})\times\frac{1}{|M^{(n)}_{s}|}.
\end{align}
We are interested in sequences of sets $\{O^{(n)}_{X_{r}}\subseteq{\cal X}^{(n)}_{r}:n\in\mathbb{N}\}$ for which:
\begin{align}
    \lim_{n\rightarrow\infty}\mu_{X_{r}|{\cal C}^{(n)}_{s}(R)}(O^{(n)}_{X_{r}})=1.
\end{align}
For any $\{O^{(n)}_{X_{r}}\subseteq{\cal X}^{(n)}_{r}:n\in\mathbb{N}\}$ there is a corresponding sequence of rates:
\begin{align}
    R^{(n)}_{X_{r}}:=\frac{1}{n}\log_{2}|O^{(n)}_{X_{r}}|.
\end{align}
There is also a subsequence $\{n_{k}:k\in\mathbb{N}\}$ of the subsequence in (\ref{subsequence}) where:
\begin{align}
    \lim_{k\rightarrow\infty}R^{(n_{k})}_{X_{r}}=R_{X_{r}}.
\end{align}
There are two cases to consider: the symbol-by-symbol regmine where $R_{Y_{r}}=R_{X_{r}}$ and the generalized block-coding regime where $R_{X_{r}}<R_{Y_{r}}$.  The block-coding regime is the one that's non-trivial, so we will discuss it.  If $R_{X_{r}}<R_{Y_{r}}$ then with high-probability according to the full $n$-length sequence, the relay is reliably converting sequences from high-rate observations to low rate transmissions.  This is only possible if the relay is looking at a significant number of observed symbols ${\cal Y}_{r}$ before encoding an equally significant number of transmitted symbols ${\cal X}_{r}$.  Significant, in this context, means some constant fraction of the total number of symbols (i.e., the full length of the codeword $n$).

This is a generalization of the original simple block-coding scheme.  In the original simple scheme, we had an input ``comb'' of observations feeding into an output comb of transmissions.  The input combs all have sequentially spaced teeth and do not overlap.  Each output comb immediately follows its input comb.  The source has the same rate accross all comb blocks.  In the generalization, the input combs have arbitrarily spaced teeth and are possibly interleaved with each other.  The combs have approximately the same number of teeth (but not necessarily exactly the same).  The source rates over different input combs might be different.  Nevertheless, this is fundamentally the same situation as the the simple block-coding scheme.  There are significant numbers of non-overlapping input combs and these can be treated as blocks.

If $R_{X_{r}}=R_{Y_{r}}$ we have a trivial case of the uncoordinated compress-forward scheme, where the destination decodes the compression after the full source sequence is transmitted.

\subsection{Inner-Bound: Achievability}
There are three achievability schemes: coordinated compress-forward, uncoordinated compress-forward, and decode-forward.  We focus on the first scheme.  The idea is to build a ``cloud'' of codewords around each compression and partially encode a message inside the compression.  This allows the source to predict the compression used by the relay.  The destination treats the compression as a normal compression and uses the normal compress-forward scheme to recover the message.  Only when decoding the compression itself (prior to decoding the message) does the destination makes use of the cloud.

\subsubsection{Codebook Generation}
We assume $B$ blocks of $n$ channel uses.  We are given a compression prior $p_{\text{cor}}(\hat{y}_{r}|x_{r})$ and a joint input distribution $p_{\text{cor}}(x_{s},x_{r})$ from ${\cal P}_{\text{cor-feasible}}$.  We get the distribution:  
\begin{align}
    \label{cloudcodebookdistribution}
    p(x_{s}|x_{r},\hat{y}_{r})&=\frac{p_{\text{cor}}(x_{s},x_{r})p_{\text{cor}}(\hat{y}_{r}|x_{r})}{\sum_{x_{s}}p_{\text{cor}}(x_{s},x_{r})p_{\text{cor}}(\hat{y}_{r}|x_{r})}
\end{align}
We define the joint distribution:
\begin{align}
    \label{jointdistributionforcoordinatedcompressforward}
    p(x_{s},x_{r},\hat{y}_{r},y_{r})=p_{\text{cor}}(y_{r}|\hat{y}_{r},x_{r})p_{\text{cor}}(x_{s},x_{r})p_{\text{cor}}(\hat{y}_{r}|x_{r}),
\end{align}
where all the terms on the right side of (\ref{jointdistributionforcoordinatedcompressforward}) are coming from ${\cal P}_{\text{cor-feasible}}$.  Now let $R_{r}=I(Y_{r};\hat{Y}_{r}|X_{r})$ with respect to the joint distribution in (\ref{jointdistributionforcoordinatedcompressforward}).  We build relay codebooks for each block.  Define:
\begin{align}
    {\cal C}^{(n,B)}_{r}(R_{r})&:=\{{\bf x}^{(n,B)}_{r}(m):m\in M^{(n,B)}_{r}\},
\end{align}
where $|M^{(n,B)}_{r}|\doteq R_{r}$ and each relay codeword is generated i.i.d according to $p(x_{r})$.  For each relay codeword ${\bf x}^{(n,B)}_{r}(m)\in{\cal C}^{(n,B)}_{r}(R_{r})$ in block $B$, we will generate a compression codebook: 
\begin{align}
   \hat{\cal C}^{(n,B)}_{r}(R_{r})|_{m_{r}}&:=\{\hat{\bf y}^{(n,B)}_{r}(\hat{m}_{r},m_{r}):\hat{m}_{r}\in\hat{M}^{(n,B)}_{r}({m_{r})}\}, 
\end{align}
where $|\hat{M}^{(n,B)}_{r}({m_{r})}|\doteq R_{r}$ and each compression codeword is generated i.i.d according to $p(\hat{y}_{r}|x_{r})$ and ${\bf x}^{(n,B)}_{r}(m)$.  For each relay codeword ${\bf x}^{(n,B)}_{r}(m)\in{\cal C}^{(n,B)}_{r}(R_{r})$ in block $B$, we build a source codebook: 
\begin{align}
    {\cal C}^{(n,B)}_{s}(R)|_{m_{r}}&:=\{{\bf x}^{(n,B)}_{s}(m_{s},m_{r}):m_{s}\in M^{(n,B)}_{s}(m_{r})\},
\end{align}
where $|M^{(n,B)}_{s}(m_{r})|\doteq R$.  The source codebook has a cloud structure so that: 
\begin{align}
    {\cal C}^{(n,B)}_{s}(R)|_{m_{r}}&:=\{{\cal C}^{(n,B)}_{s}(R_{s})|_{(\hat{m}_{r},m_{r})}:\hat{m}_{r}\in\hat{M}^{(n,B)}_{r}(m_{r})\},
\end{align}
where $|\hat{M}^{(n,B)}_{r}(m_{r})|\doteq R_{r}$ and $R_{s}=R-R_{r}$.  The source cloud codebooks are built around a relay-codeword-compression-codeword pair $({\bf x}^{(n,B)}_{r}(m_{r}),\hat{\bf y}^{(n,B)}_{r}(\hat{m}_{r},m_{r}))$ according to the distribution in (\ref{cloudcodebookdistribution}) to get:
\begin{align}
    {\cal C}^{(n,B)}_{s}(R_{s})|_{(\hat{m}_{r},m_{r})}&:=\{{\bf x}^{(n,B)}_{s}(m_{s},\hat{m}_{r},m_{r}):m_{s}\in M^{(n,B)}_{s}(\hat{m}_{r},m_{r})\},
\end{align}
where $|M^{(n,B)}_{s}(\hat{m}_{r},m_{r})|\doteq R_{s}$.  Every message $m\in M^{(n,B)}_{s}(m_{r})$ maps uniquely to an index pair $(m_{s},\hat{m}_{r})$ where $\hat{m}_{r}\in\hat{M}^{(n,B)}_{r}(m_{r})$ and $m_{s}\in M^{(n,B)}_{s}(\hat{m}_{r},m_{r})$. 

\subsubsection{Encoding}
First, we look at the source.  Suppose the source wants to send the message $m\in M^{(n,B)}_{s}(m_{r})$ in block $B$.  This message corresponds to the index pair $(m_{s},\hat{m}_{r})$ where $\hat{m}_{r}\in\hat{M}^{(n,B)}_{r}(m_{r})$ and $m_{s}\in M^{(n,B)}_{s}(\hat{m}_{r},m_{r})$.  The source sends the codeword ${\bf x}^{(n,B)}_{s}(m_{s},\hat{m}_{r},m_{r})$ in block $B$.  

Now we look at the relay.  The relay knows the compression codeword $\hat{\bf y}^{(n,B-1)}_{r}(\hat{m}_{r},m_{r})$ and the relay codeword ${\bf x}^{(n,B-1)}_{r}(m)$ sent in block $B-1$.  In block $B$, the relay sends ${\bf x}^{(n,B)}_{r}(\hat{m}_{r})$.

\subsubsection{Decoding}
First, we look at the relay.  At the end of block $B$, the relay finds the compression codeword $\hat{\bf y}_{r}$ that satisfies the joint-typicality check:
\begin{align}
    (\hat{\bf y}_{r},{\bf x}^{(n,B)}_{r},{\bf y}^{(n,B)}_{r})\in T^{(n)}_{*}(Y_{r},\hat{Y}_{r}|X_{r}).
\end{align}
Now we look at the destination.  Define the index recovery mapping 
\begin{align}
    \hat{I}^{-1}_{r}:\hat{\cal C}^{(n,B)}_{r}(R_{r})|_{m_{r}}\rightarrow\hat{M}^{(n,B)}_{r}(m_{r})
\end{align}
At the end of block $B$, the destination recovers the compression $\hat{\bf y}^{(n,B-1)}_{r}$ sent by the relay in block $B$ and the source in block $B-1$, via the following joint typicality checks:
\begin{align}
    ({\bf x}_{r}(I^{-1}(\hat{\bf y}_{r})),{\bf y}^{(n,B)}_{d})&\in T^{(n)}_{*}(X_{r},Y_{d}),\\
    (\hat{\bf y}_{r},{\bf x}^{(n,B-1)}_{r},{\bf y}^{(n,B-1)}_{d})&\in T^{(n)}_{*}(\hat{Y}_{r},Y_{d}|X_{r}).
\end{align}
Finally, the destination finds the source codeword that satisfies the following joint-typicality checks:
\begin{align}
    ({\bf x}_{s},{\bf x}^{(n,B-1)}_{r},\hat{\bf y}^{(n,B-1)}_{r})&\in T^{(n)}_{*}(X_{s},\hat{Y}_{r}|X_{r}),\\
    ({\bf x}_{s},{\bf x}^{(n,B-1)}_{r},\hat{\bf y}^{(n,B-1)}_{r},{\bf y}^{(n,B-1)}_{d})&\in T^{(n)}_{*}(X_{s},Y_{d}|X_{r},\hat{Y}_{r}).
\end{align}

\subsection{Analysis of Error Probability}
First we look at the source.  An error occurs if the source does not know the relay codeword sent in block $B$.  This happens if the relay finds a different compression $\hat{\bf y}^{(n,B-1)}_{r}$ from the one used to generate the source codeword.  The probability of this error event vanishes if $R_{r}<I(\hat{Y}_{r};Y_{r}|X_{r})$.  

Now we look at the destination.  An error occurs if the destination decodes the wrong compression codeword at the end of block $B$.  The probability of this error event vanishes if $R_{r}<I(X_{r};Y_{d})+I(\hat{Y}_{r};Y_{d}|X_{r})=I(\hat{Y}_{r},X_{r};Y_{d})$.  An error also occurs if the destination decodes the wrong source codeword at the end of block $B-1$.  The probability of this error event vanishes if $R<I(X_{s};\hat{Y}_{r}|X_{r})+I(X_{s};Y_{d}|X_{r},\hat{Y}_{r})=I(X_{s};\hat{Y}_{r},Y_{d}|X_{r})$.


\section{Conclusion}
\label{sec:conclusion}
The capacity of the relay channel was found using properties of typical sequences.  We derived convergence properties of sequences of conditionally typical sets anchored to arbitrary sequences of codewords.  We developed a sphere-packing argument for the converse to the capacity of a point-to-point channel.  We used the compactness of the simplex in the proof.  We introduced the capacity of the relay channel characterized in terms of three schemes: decode-forward, coordinated compress-forward, and uncoordinated compress-forward.  We presented a dichotomy in the compress-forward space; there is only one degree-of-freedom between the input distribution and the reverse-compression-channel distribution.  We show this dichotomy creates a zero-one circuit consistency law where coordination is only possible if the forward channel and reverse channel are aligned.  The proof of the capacity establishes a compression-cut-set bound for compress-forward schemes.  We use the compactness of the interval $[0,|{\cal X}_{r}|]$ and a sphere-packing argument to get the codeword rate contributing to the mass of a compression.  The same compactness argument reveals a symbol-by-symbol regime and a generalized-block-coding-regime.  We introduced coordinated compress-forward to show achievability.

\ifCLASSOPTIONcaptionsoff
  \newpage
\fi



\bibliographystyle{IEEEtran}
\bibliography{bibtex/bib/IEEEexample}

\begin{thebibliography}{1}
\providecommand{\url}[1]{#1}
\csname url@samestyle\endcsname
\providecommand{\newblock}{\relax}
\providecommand{\bibinfo}[2]{#2}
\providecommand{\BIBentrySTDinterwordspacing}{\spaceskip=0pt\relax}
\providecommand{\BIBentryALTinterwordstretchfactor}{4}
\providecommand{\BIBentryALTinterwordspacing}{\spaceskip=\fontdimen2\font plus
\BIBentryALTinterwordstretchfactor\fontdimen3\font minus \fontdimen4\font\relax}
\providecommand{\BIBforeignlanguage}[2]{{%
\expandafter\ifx\csname l@#1\endcsname\relax
\typeout{** WARNING: IEEEtran.bst: No hyphenation pattern has been}%
\typeout{** loaded for the language `#1'. Using the pattern for}%
\typeout{** the default language instead.}%
\else
\language=\csname l@#1\endcsname
\fi
#2}}
\providecommand{\BIBdecl}{\relax}
\BIBdecl

\bibitem{Van_Der_Meulen_1971}
E.~C. {Van Der Meulen}, ``{T}hree-terminal communication channels,'' \emph{Advances in Applied Probability}, vol.~3, no.~1, pp. 120--154, 1971.

\bibitem{Kramer2001}
G.~Kramer, ``Genie-aided outer bounds on the capacity of interference channels,'' in \emph{Proceedings. 2001 IEEE International Symposium on Information Theory (IEEE Cat. No.01CH37252)}, 2001, pp. 103--.

\bibitem{Wu2009}
\BIBentryALTinterwordspacing
X.~Wu and L.-L. Xie, ``{Asymptotic Equipartition Property of Output when Rate is above Capacity},'' 2009. [Online]. Available: \url{https://arxiv.org/abs/0908.4445}
\BIBentrySTDinterwordspacing

\bibitem{Cover1979}
T.~M. Cover and A.~El~Gamal, ``Capacity {T}heorems for the {R}elay {C}hannel,'' \emph{IEEE Transactions on Information Theory}, vol.~25, no.~5, pp. 572--584, 1979.

\end{thebibliography}
\end{document}